\documentclass{article}
\usepackage{arxiv}

\usepackage{amsmath,amssymb,graphicx,pslatex}

\usepackage[T1]{fontenc}
\usepackage{hyperref}

\usepackage{paralist}
\usepackage[bottom]{footmisc}
\usepackage[linesnumbered,lined,boxruled, boxed, vlined, noend]{algorithm2e}
\usepackage{amsthm}

\usepackage{color}
\hypersetup{colorlinks=false, pdfborder={0 0 0}}

\definecolor{dgreen}{rgb}{0,.5,.0}
\definecolor{purple}{rgb}{.5,0,.5}
\definecolor{green}{rgb}{0,1,0}
\definecolor{lgreen}{rgb}{.5,1,.5}

\newcommand{\elll}{\ell^{\rightarrow}}
\newcommand{\ellrr}[1]{\ell_{#1}^{\leftarrow}}

\newcommand{\ESEll}[1]{{\mathcal{S}_{#1}^{\rightarrow}}}

\newcommand{\ESErr}[1]{{\mathcal{S}_{#1}^{\leftarrow}}}
\newcommand{\Wl}{{\mathcal{W}^{\rightarrow}}}
\newcommand{\Wrigth}{{\mathcal{W}^{\leftarrow}}}

\DeclareMathOperator{\len}{len}

\newtheorem{theorem}{Theorem}
\newtheorem{lemma}{Lemma}
\newtheorem{corollary}{Corollary}
\newtheorem{claim}{Claim}
\newtheorem{definition}{Definition}

\newtheorem{problem}{Problem}
\newtheorem{strategy}{Strategy}

\newtheorem{observation}{Observation}

\newcommand{\Section}[1]{\section{#1}}

\newcommand{\Subsection}[1]{~~\\\noindent {\bf #1.}~~}

\author{Oscar Morales-Ponce \\
  Department of Computer Engineering and Computer Science\\
  California State University Long Beach\\
  Long Beach, CA, 97840 \\
  \texttt{oscar.moralesponce@csulb.edu} \\
}

\date{}                      

\title{Optimal Patrolling of High Priority Segments While \\ Visiting the Unit Interval with a Set of Mobile Robots}

\begin{document}

\maketitle

\begin{abstract}
Consider a region that requires to be protected from unauthorized penetrations.  The border of the region, modeled as a unit line segment, consists of high priority segments that require the highest level of protection separated by low priority segments that require to be visited infinitely often. 
We study the problem of patrolling the border with a set of $k$ robots. The goal is to obtain a strategy that minimizes the maximum idle time (the time that a point is left unattended) of the high priority points while visiting the low priority points infinitely often. We use the concept of single lid cover (segments of fixed length) where each high priority point is covered with at least one lid, and then we extend it to strong double-lid cover where each high priority point is covered with at least two lids, and  the unit line segment is fully covered. Let $\lambda_{k-1}$ be the minimum lid length that accepts a single $\lambda_{k-1}$-lid cover with $k-1$ lids and $\Lambda_{2k}$ be the minimum lid length that accepts a strong double $\Lambda_{2k}$-lid cover with $2k$ lids. We show that $2\min(\Lambda_{2k}, \lambda_{k-1})$ is the lower bound of the idle time when the max speed of the robots is one. To compute $\Lambda_{2k}$ and $\lambda_{k-1}$, we present an algorithm with time complexity $O(\max(k, n)\log{n})$ where  $n$ is the number of high priority sections. Our algorithm improves by a factor of $\min(n,k)$ the previous $O(kn\log n)$ running time algorithm. For the upper bound, first we present a strategy with idle time $\lambda_{k-1}$ where one robot covers the unit line, and the remaining robots cover the lids of a single $\lambda_{k-1}$-lid cover with $k-1$ lids. Then, we present a simple strategy with idle time $3\Lambda_{2k}$ that splits the unit line into not-disjoint $k$ segments of equal length that robots synchronously cover, i.e., reaching the leftmost and rightmost point simultaneously. Then, we present a complex strategy that split the unit line into $k$ non-disjoint segments that robots asynchronously cover. We show that combining strategies one and two attain an approximation of $1.5$ the optimal idle time and combining strategy one and third attain optimal idle time.
\end{abstract}

\maketitle

\Section{Introduction}
Consider a region that requires to be protected from  unauthorized penetrations. A team of robots can patrol (perpetually move along) the border of the region looking for intruders. The problem is known as the patrolling problem and has been extensively studied. See for example \cite{chuangpishit2018patrolling,collins2013optimal,czyzowicz2017patrolmen,czyzowicz2011boundary,czyzowicz2016patrolling}. Most of the previous works assume that all the points in the border have the same priority to be visited. However, the border may consist of sections with different priority of patrolling such as  static guards protecting some sections of the border that are required to be visited regularly to detect points of failure. On the other side, every high priority point is required to be visited as often as possible. What strategy must the robots follow to give the maximum protection to the points of the high priority sections while visiting  infinitely often every point of the border? To answer this question, we consider the \emph{idle time} to measure the efficiency of a strategy. Intuitively, the idle time of a given strategy with $k$ robots measures the maximum period that any high priority point remains unvisited. We are interested in providing strategies that achieve optimal idle time.

\Subsection{Results of the Paper and Contributions}
In this paper, we study the problem of minimizing the idle time of high priority sections in a unit segment while low priority segments are visited infinitely often using $k$ mobile robots. Let $I^*$ denote the minimum idle time that any strategy can attain. We use the concept of lids (segment of fixed length)  to characterize the problem. We use the single lid cover proposed in~\cite{collins2013optimal} where each high priority point is covered by at least one lid, and then we extend to strong double lid cover where every high priority point is covered by at least two lids and the unit segment is fully covered. Unlike strong double lid covers a double lid cover would not need to fully cover the unit segment. Let $\Lambda_{2k}$ denote the minimum lid length such that the unit segment admits a strong double lid cover with $2k$ lids and let $\lambda_{k-1}$ denote the minimum lid length such that the unit segment admits a single lid cover with $k-1$ lids. We show that $I^* \geq 2\min(\lambda_{k-1}, \Lambda_{2k})$. The proof is based on finding $k+1$ points at distance at least $\min(\lambda, \Lambda)$ where $k$ points are high priority. 

Then, we present a $O(\max(k, n)\log n )$ running time algorithm that determines the minimum lid length of a strong double lid cover as well as the single lid cover where $n$ is the number of high priority segments and $k$ is the number of robots. First, we present a  $O(\max(k, n))$ running time algorithm that given a length $l$ it decides whether the unit segment accepts a  strong double cover with lid length $l_{min}\leq l$. Then, we use a binary search to find the optimal value. The algorithm improves by a factor of $\min(k,n)$ the running time of $O(kn\log n)$ in \cite{collins2013optimal}.

Regarding the upper bound, we present three strategies:  Given a single cover with $k-1$ lids of length $\lambda_{k-1}$ the first strategy with idle time $2\lambda_{k-1}$ assigns a robot to each lid, and one robot covers the unit segment. The second and third strategies rely on a strong double lid cover with 
$2k$ lids of length $\Lambda_{2k}$. In particular, the second strategy with idle time of at most $3\Lambda_{2k}$ assigns robots to $k$ equal segments of length at most $2\Lambda_{2k}$. Each robot is, then, assigned to one segment where they move synchronously back and forth. The third strategy with idle time $2  \Lambda_{2k}$ assigns each robot to two lids. Robots cover one of the lids and asynchronously move to the other lid. 

\Subsection{Organization of the Paper}
The remaining of the paper is organized as follows.  In Section~\ref{sec:model} we present the model, define the problem formally and present an illustrative example with one high priority section and two robots that provide the insights of the techniques to compute the lower and upper bounds. The related work is then presented in Section~\ref{sec:relatedwork}. In Section~\ref{sec:lid} we introduce the concept of strong double lid cover and provide important properties that we use in later sections. Then, the properties of the strong double and the single lid cover are used to provide the lower bounds in Section~\ref{sec:lower}.  An algorithm with running time  $O(\max(k,n)\log n )$ is presented in Section~\ref{sec:algorithm} where $k$ is the number of robots and $n$ is the number of high priority sections. We continue in Section~\ref{sec:upper} providing the strategies that attain optimal idle time and a strategy that approximates within $1.5$ the optimal idle time. The conclusion is then presented in Section~\ref{sec:conclusion}. 

\section{Model and Problem Statement}\label{sec:model}

Without loss of generality, we model the border as a segment of unit length $C=[0,1]$. The segment is partitioned with two subsets $H$ and $L$ where $H$ represents the sections of ``high priority", and $L$ the sections of ``low priority."  We take $H$ to be a finite union of closed intervals. 

We consider $k$ identical robots with a maximum speed of one. In our model, we assume that the acceleration is infinite. Hence, robots can change speed instantly. 

Each robot $r_i$ follows a continuous function $f_i(t)$ that defines the position of the robot $r_i$  in the unit interval $C$ at time $t$. A \emph{strategy} consists of $k$ continuous  functions. The idle time of the strategy is defined as the minimum-maximum time that any point in $H$ remains unvisited. Formally: 

\begin{definition}[Idle time]
Let $A$ be a strategy for $k$ robots. The \emph{idle time} induced by a strategy ${A}$ at a point $x \in H$, denoted by $I_{x} (A)$, is the maximum time interval that  $x$ remains unvisited by any robot while the low priority points are visited infinitely often: $$I_{x} (A) = \sup_{\{0\leq T_1 < T_2\colon \forall i \forall t\in (T_1,T_2)\ f_i(t) \neq x \}} (T_2 - T_1).$$ Let $I_A  = \sup_{\forall x \in H} (I_{x} (A))$ denote the idle time of the strategy $A$.
\end{definition}

We are interested in determining the minimum idle time that any strategy can attain. Formally, the problem that we address is.

\begin{problem} 
Given a partition $(H,L)$ of the unit interval $C=[0,1]$ and a set of $k$ robots with the same maximum speed of one, determine the optimal idle time that any strategy $A$ can attain, i.e., $$I_k^* = \inf_{\forall A}(I_A).$$
\end{problem}

\subsection{Two Robots One High Priority Segment} \label{tworobots}
Consider as an example a single high priority segment with two robots $r_1$ and $r_2$. Refer to Figure~\ref{fig:cap1}. Let $a< b$ be the left and right points of the high priority segment. Without loss of generality assume that $a \leq 1-b$.

\begin{figure}[htbp]
\centering
\includegraphics[scale=.45]{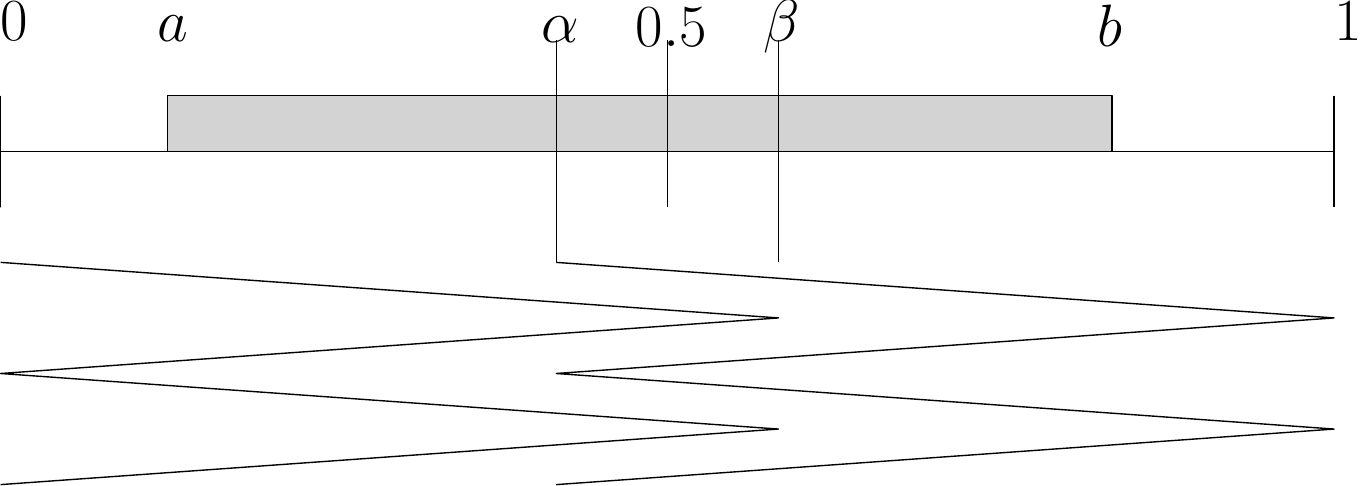}
\caption{One priority segment.}
\label{fig:cap1}
\end{figure}

\emph{Proposed Strategy}: $r_1$ goes back and forth between $0$ and $\beta$ (with an appropriately chosen $\beta\ge 0.5$), and $r_2$ goes back and forth between $\alpha = 1 - \beta$ and $1$. We are mostly concerned with the idle time at the following points: 
\begin{itemize}
\item Point $a$.  When $r_1$ visits $a$ after visiting  $\beta$, the idle time is $2(\beta-a) = 2(1-\alpha-a)$), 
\item Point $b$.  When  $r_2$ visits $\alpha$ after visiting $b$, the idle time is $2(b - (1 - \beta))= 2(b -\alpha)$.
\item Point $\alpha - \epsilon$. When $r_1$ visits $\alpha - \epsilon$ after visiting $0$. The idle time  is $2(\alpha - \epsilon)$. Taking the supremum, the idle time is $2\alpha$.
\item  Point $\beta + \epsilon$. When $r_2$ visits $\beta + \epsilon$ after visiting $1$. The idle time is $2(1 - (\beta  +  \epsilon)) =2(\alpha - \epsilon)$. Taking the supremum, the idle time is $2\alpha$.
\end{itemize}

Observe that since $a < 1 - b$,  the idle time $2(1-\alpha - a )$ is at least $2(b -\alpha)$. Therefore, we can calculate $\alpha = \frac{1-a}{2}$ by setting $2(1-\alpha -a ) = 2\alpha$   and factorizing $\alpha$. Thus, the idle time of the strategy is $1- a$. Note that the idle time of any other high priority point is not greater than $1- a$. A second strategy is that $r_1$ covers $[a,b]$ meanwhile $r_2$ covers $[0,1]$. Clearly, the idle time of the strategy is $2(b-a)$.

\begin{theorem}
The optimal idle time with two robots when there is a single high priority segment $[a,b]$ such that $a \leq 1-b$ is $\min(2(b-a), 1-a).$
\end{theorem}

\begin{proof}
We show that there is no strategy that attains idle time $\min(2(b-a), 1-a) - \epsilon$. Consider first when $2(b-a) \leq 1-a$. Therefore, rearranging we have that $b-a \leq 1-b$. Consider the time $t_0$ when one robot visits $1$.  Such a time exists since the low priority points are visited infinitely often. However, since $b-a \leq 1-b$, the same robot cannot visit $b$ in the time interval  $[t_0 - (b-a-\epsilon/2), t_0 + (b-a-\epsilon/2)]$. Therefore, the other robot must visit $b$ at time $t_1 \in [t_0 - (b-a-\epsilon/2), t_0 + (b-a-\epsilon/2)]$. However, $a$ remains being unvisited for $2(1-b) - \epsilon$ which contradicts the assumption.

Consider now when $1-a < 2(b-a)$. Therefore, rearranging we have that $b > (1+a)/2$. Take the points $a, (1+a)/2$ and $1$. Since  $b > (1+a)/2$, $(1+a)/2$ is a high priority point. Similar as before, there is a time that one robot visits $1$, say $t_0$.  However, the same robot cannot visit $b$ in the time interval  $t_0 -((1-a)/2 - \epsilon/2), t_0 + ((1-a)/2-\epsilon/2)$. Therefore, the other robot must visit the point $(1+a)/2$ at time $t_1 \in [t_0 -(1-a)/2 + \epsilon/2, t_0 + (1-a)/2-\epsilon/2]$. However, $a$ remains being unvisited  for $1-a + \epsilon$ time which contradicts the assumption.
\end{proof}

Let $\alpha = (1-a)/2$ and $\beta =  (1+a)/2$, we have the following observation.

\begin{observation}\label{obse:1}
The segments of length $\alpha$ (or $\alpha$-lids thereafter)  $[0, \alpha]$, $[a,  \beta]$, $[\alpha, 2\alpha]$ and $[\beta, 1]$ double cover $[a,b]$, i.e., every point in  $[a,b]$ is in at least two different lids and the union of the lids covers the unit line segment. 
\end{observation}

In the sequel, we generalize the previous results and use Observation~\ref{obse:1} to characterize the optimal idle time for any number of high priority segments and any number of robots.

\section{Related Work}\label{sec:relatedwork}

Different variations of the patrolling problem have been studied recently. A closely related version was studied in~\cite{collins2013optimal} where the border is divided into two different types of segments. Namely, the vital sections that are required to be visited with the minimum time and the neutral sections that are not required to be visited, but robots can traverse to reach vital sections. They provide an optimal strategy for the unit segment and the ring. The problem studied in this paper requires that all points must be visited infinitely often unlike~\cite{collins2013optimal}.

Another closely related problem was recently studied in \cite{chuangpishit2018patrolling} where two robots are required to patrol a set of points on a unit segment. In that paper, points can be assigned different priorities. The problem asks to find strategies that guarantee that the maximum time that a point remains unvisited is at most the given priority. The authors provide a $\sqrt{3}$-approximation algorithm. A similar problem was studied in \cite{gkasieniec2017bamboo} where a single robot is required to visit a set of $n$ points with priorities. These priorities are updating at a steady rate, and the problem asks to find a strategy that minimizes the maximum priority ever observed. The authors study two different variants of the model and provide the upper bound.

Patrolling without priorities have been studied in different contexts. For example, in \cite{czyzowicz2016patrolling}, the authors studied the problem of patrolling the edges of a geometric tree. The problem asks to minimize the time that any point in the edges is left unvisited (idle time). They show an off-line strategy that attains optimal idle time for any number of robots. In \cite{czyzowicz2017patrolmen} the authors study the patrolling problem in a geometric graph where robots are unreliable. They propose an optimal algorithm for the line segments that are then used for the general graphs. Giving an Eulerian graph $G=(V,E)$, they show an idle time of $(f+1)|E|/k$ where $k$ is the number of robots, $f \leq k$ is the number of robots that can fail and $|E|$ is the sum of the lengths of the edges of $G$. They show that the general problem is NP-hard even with three robots where at most one can fail. In \cite{czyzowicz2011boundary} the authors study the patrolling problem on a unit cycle where robots have different speed. They obtain optimal results when the number of robots is 2, 3, and 4. Similarly in \cite{kawamura2012fence} the authors consider the same problem in a unit segment with $k$ robots. They show an optimal strategy that attains idle time of $(v_1 + v_2, + v_3)/2$ for three robots where $v_i$ is the max speed of robot $i$.

Patrolling with visibility has been studied in \cite{czyzowicz2014patrolling} where the authors consider that each robot has visibility $r_i$. They propose an optimal algorithm to patrol the unit interval when the robots have the same maximum speed. A survey of the patrolling problem was recently published in~\cite{czyzowicz2019patrolling}. The authors discuss strategies for various number of robots with different capabilities in different domains. 

\section{Lid Covers}~\label{sec:lid}
In this section, we introduce the concept of lids that will be used to derive the lower bound and to describe the strategies that the robots must follow. In  \cite{collins2013optimal} the authors study the patrolling problem when a continuous rectifiable curve contains vital sections and neutral sections. Unlike our problem, the neutral sections are not required to be visited. To describe the trajectory, they introduce the concept of lids. More formally, an \emph{$l$-lid} is a closed interval of length $l$ intersecting $C$. A lid cover is then defined as:

\begin{definition}\label{def:weakdouble}[Single $l_k$-lid cover.]
An \emph{ $l_k$-lid cover} is a set $\{\ell_1,\ell_2\dotsc,\ell_k\}$ of $k$ $l_k$-lids such that every high priority point in $p\in H$ is covered by at least one lid, i.e., for every point $p \in H$ there exists $\ell_i$ such that $p\in \ell_i$. 
\end{definition}

Let $Left(\ell)$ and $Right(\ell)$ be the leftmost and rightmost  point in the lid $\ell$. For any lid $\ell$, define the leftmost \emph{high-priority point} of $\ell$ by
$$L(\ell) =
\begin{cases}
\inf(\ell \cap H)  & \mbox{  if } \ell \cap H \not= \emptyset\\
Right(\ell) &\mbox{ if } \ell \cap H = \emptyset\\
\end{cases}
 $$ 
Analogously, we define $R(\ell)$.
 
We now introduce the concept of a double lid cover of $C$ that is used as the critical tool to prove the lower bounds and to provide the strategy that attains the optimal idle time. Essentially, the strong double $l$-lid cover requires that at least two lids cover each high priority segment and $C$ is fully covered.  The following definition formalizes the concept of strong double $l$-lid cover.

\begin{definition}\label{def:strongdouble}[Strong double $l_k$-lid cover.]
A strong double $l_k$-lid cover is a set $\{\ell_1,\ell_2\dotsc,\ell_k\}$ of $l_k$-lids such that the unit segment $C$ is fully covered and every high priority point in $p\in H$ is covered by at least two distinct lids, i.e., $\bigcup_{i=1}^k\ell_i = C$ and there exist $\ell_i$ and $\ell_{j\not=i}$ such that $p\in \ell_i\cap \ell_j$ for all $p\in H$.
\end{definition}

We generally write $\mathcal{W}_k(l)$ (respectively $\mathcal{S}_k(l)$) for an arbitrary single lid cover (respectively strong double lid cover) with $k$ lids of length $l$. Let $\mathbb{W}_k(l)$ (respectively $\mathbb{S}_k(l)$) denote the set of single lid covers (strong double lid covers). Next, we show the existence of a strong double and single cover with optimal length. For convenience, we  order the lids in $\mathcal{W}_k(l)$ and $\mathcal{S}_k(l)$ from left to right, i.e., $Left(\ell_i) \leq Left(\ell_{i+1})$ for all $i< k$. Ties are broken arbitrarily.

\begin{definition} \label{def:strong}
The \emph{$k$-single, $k$-strong lower bound} of the lid length $l$ is $\lambda_k = \inf \{l\mid \mathbb{W}_k(l) \neq \emptyset\}$, $\Lambda_k = \inf \{l\mid \mathbb{S}_k(l) \neq \emptyset\}$, respectively.
\end{definition}
We omit the subscript $k$ if it is understood from the context. 

\begin{lemma}\label{thm:existenceOfMin}
Both $\{l\mid \mathbb{W}_k(l) \neq \emptyset\}$ and $\{l\mid \mathbb{S}_k(l) \neq \emptyset\}$  have a minimum.
\end{lemma}

\begin{proof}
We focus on strong double lid covers; the case for single lid cover is similar.  For any strong double lid cover $\mathcal{C} = \{\ell_1,\ell_2\dotsc,\ell_k\}$, let $\sigma(\mathcal{C}) = a_1,a_2,\dotsc,a_{2k}$ be the sequence of (not-necessarily distinct) endpoints of lids arranged in the order the lids appear in the sequence $\ell_1,\dotsc,\ell_k$. Thus for each $j$, $a_{2j-1}$ is the leftmost point of a lid whose rightmost point is $a_{2j}$.

For each $n$, let $\mathcal{C}^n$ be a strong double lid cover with lid length at most $\Lambda_k+1/n$. Write $\sigma(\mathcal{C}^n) = a_1^n,a_2^n,\dotsc,a_{2k}^n$. Now $(a_1^n)$ is a bounded sequence and by the Bolzano-Weierstrass theorem has a convergent subsequence $a_1^{j^1_n}$. Then we apply the Bolzano-Weierstrass theorem to $(a_2^{j^1_n})$ to get a convergent subsequence $a_2^{j^2_n}$, where $j^2_n = j_1^{k_n}$ for some increasing sequence $k$. We continue until we have a $(j^{2k}_n)$, which we abbreviate as $(r_n)$, such that $a_i^{r_n}$ converges for each $i$. For each $i$, let $\lim_{n\to \infty} a_i^{r_n} = a_i$.

We show that the sequence $a_1,a_2,\dotsc, a_{2k}$ is the sequence $\sigma(\mathcal{C})$ for some strong double lid cover $\mathcal{C}$ with lid-length $\Lambda_k$. Note that $a_{2j+2}-a_{2j+1} = \lim_{n\to \infty} a_{2j+2}^{r_n} - \lim_{n\to \infty} a_{2j+1}^{r_n} = \lim_{n\to \infty} (a_{2j+2}^{r_n} - a_{2j+1}^{r_n}) = \Lambda_k$. So the sequence $a_i$ consists of lids each of length $\Lambda_k$.

To see that every point is covered, suppose toward a contradiction that there is a point $z$ that is not in any lid. Let $\epsilon$ be the smallest distance from $z$ to any of the $a_i$. Then, there is an $N$ such that each $a_i^{r_n}$ is within $\epsilon$ of $a_i$ for each $i$ and $n\ge N$. It is then easy to see that $z$ would have to be left out of each of the $\mathcal{C}_n$ for $n\ge N$, a contradiction.

To show that every priority region is double covered, we use a similar strategy. Suppose $z$ is a high priority point not double covered. Let $S$ be the set of indices for $i$ such that $a_i$ is an endpoint of a lid not containing $z$ (there are at least $2k-2$ such $i$). Let $\epsilon$ be the smallest distance from $z$ to any of the $a_i$ for $i\in S$. Then, there is an $N$ such that each $a_i^{r_n}$ is within $\epsilon$ of $a_i$ for each $i\in S$ and $n\ge N$. It is then easy to see that $z$ would have to be left out of each $k-1$ lids of $\mathcal{C}_n$ for $n\ge N$ whose endpoints have indices in $S$, a contradiction. 
\end{proof}

Let $\mathcal{W}_k$ denote a single lid cover with $k$  lids of optimal length $\lambda_k$. Throughout the paper we consider only optimal single  lid covers. A \emph{block} of $\mathcal{W}_k$ is a tuple $B = \{\ell_{a(1)}, \ell_{a(2)}, ...\ell_{a(b)}\}$ of lids such $Right(\ell_{a(i)})  = Left(\ell_{a(i+1)})$ for all $i \in [1,b-1]$. Let $Left(B) = Left(\ell_{a(1)})$ and $Right(B) = Right(\ell_{a(b)})$. Let
\[
L(B) = \begin{cases}
\inf(H\cap (\bigcup_{k=1}^{b}\ell_{a(k)}))  & \mbox{  if } H\cap \bigcup_{k=1}^{b}\ell_{a(k)} \not= \emptyset\\
Right(\ell_{{a(b)}}) &\mbox{ if } H\cap \bigcup_{k=1}^{b}\ell_{a(k)} = \emptyset\\
\end{cases}
\]
and define $R(B)$ analogously. We define a \emph{critical block} $B$ of $\mathcal{W}_k$ as a block such that $Left(B)$ is the left endpoint of a high priority segment and $Right(B)$ is the right endpoint of a high priority segment.

Let $\mathbb{W}$ be an optimal single $\lambda_k$-lid cover with $k$ lids. Observe that we can shift all lids that the leftmost point does not cover a high priority point to the right until the leftmost point reaches a high priority point. Let $\Wl_k$ denote such a single lid cover.  Analogously, we define $\Wrigth_k$ where the rightmost point of every lid covers a high priority point.

\begin{lemma}\label{lem:criticalblock}
Every $\Wl_k$  has a  critical block.
\end{lemma}

\begin{proof}
Assume by contradiction that $\Wl_k$  does not have a critical block. Therefore, $Right(B_i) - R(B_i) > 0$ for each block $B_i$. Let $$d_{min} = \min_{\forall B_i}\left(\frac{Right(B_i) - R(B_i)}{|B_i|}\right).$$ We can shrink each lid by  $d_{min}$   without uncovering any high priority segment, contradicting that the lid length of $\Wl_k$ is optimal.
\end{proof}

Let $\Lambda_{k}$ be the minimum lid length of a strong double lid cover with $k$ lids. Let $\mathcal{S}_k$ denote the optimal double $\Lambda_{k}$-lid cover. Throughout the paper, we merely consider $\mathcal{S}_k$. We observe that in $\mathcal{S}_k$ some priority points may be triple covered. We depict this situation in Figure~\ref{ex:tripleCov}. 
The next lemma guarantees that the high priority points are covered exactly by two lids unless they are endpoints of the lids, in which case they can be covered by at most three endpoints. 

\begin{figure}[htbp]
\centering
\includegraphics[scale=1.0]{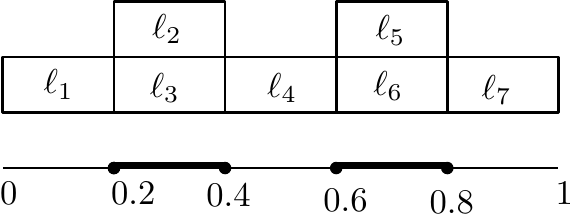}
\caption{Priority segments $[0.2,0.4]$ and $[0.6,0.8]$. Observe that  the priority points $0.2$,$0.4$, $0.6$ and $0.8$ are  triply covered.}
\label{ex:tripleCov}
\end{figure}

\begin{lemma}\label{thm:rightshiftlaxcover}
Every double $\Lambda_k$-lid cover can be transformed into a double $\Lambda_k$-lid cover satisfying:
\begin{enumerate}
\item There are no points covered more than three times, 
\item Any triply covered point is one of the endpoints of two distinct lids, and
\item For any overlapping segment that has a high-priority point, the leftmost point of the overlapping segment is a high-priority point. 
\end{enumerate}
\end{lemma}

\begin{proof}
The proof is by construction.  Place lids $\ell_1,\ell_2,\dotsc,\ell_{k}$ such that:
\begin{itemize}
\item $\ell_1 = [0, \alpha_1]$ is the leftmost lid such that $\alpha_1 = \Lambda_k$.
\item $\ell_2 = [\alpha_2, \alpha_2+\Lambda_k]$ such that $\alpha_2 =  L(\ell_1)$.
\end{itemize}
For each $i\in \{3,\dotsc,k\}$, let  
\begin{itemize}
\item $\ell_i = [\alpha_i,\alpha_i+\Lambda_k]$ where 
$$\alpha_i =  
\begin{cases}
L(\ell_{i-1}\setminus\ell_{i-2}) & \mbox{ if } i < k \\ 
 \min\{1, Right(\ell_{i-1})\} &  \mbox{ otherwise }
\end{cases}
$$ 

\end{itemize}
By construction, the resulting double $\Lambda_k$-lid cover has the desired conditions. 
\end{proof}

We call a strong double $\Lambda_k$-lid cover that satisfies  conditions 1-3 of Lemma~\ref{thm:rightshiftlaxcover} a \emph{left  shifted  double $\Lambda_k$-lid cover} and denoted as $\ESEll{k}$. Observe that by reversing the order we can analogously obtain a \emph{right  shifted  double $\Lambda_k$-lid cover}  denoted as $\ESErr{k}$. Observe that $\ESEll{k}$ and  $\ESErr{k}$ are unique.

In the context of strong double covers, the endpoints $0$ and $1$ of the segment $C$ are of particular importance, and we define $H^* = H\cup \{0,1\}$. A block $B$ of a double cover is defined the same way as for single lid covers, and $Left(B)$ and $Right(B)$ are also defined in the same way; let $\mathcal{B}$ be the set of all  blocks in  $\ESEll{k}$. We define an order on $\mathcal{B}$ by setting $B_i<B_j$ if $Left(B_i) < Left(B_j)$. Ties are arbitrarily broken. We define $L^*(\ell)$, $R^*(\ell)$, $L^*(B)$, and $R^*(B)$ the same way we did for $L(\ell)$, $R(\ell)$, $L(B)$, $R(B)$, but with $H$ replaced with $H^*$.

Next, we show a property that we use for the lower bound in Section~\ref{sec:lower}. A block $B$ is \emph{critical} in $\ESEll{k}$ if $Left(B)$ is a left endpoint and $Right(B)$ is a right endpoint of a segment of $H$. 

\begin{lemma}\label{lem:dcritical}
$\ESEll{k}$  has a critical block. 
\end{lemma}

\begin{proof}
Let $\mathcal{B}$ be the set of all  maximal blocks $B$ in  $\ESEll{k}$. Note that every lid of $\ESEll{k}$ is either contained in some block in $\mathcal{B}$ or has a left endpoint of $1$. Assume by contradiction that $\ESEll{k}$ does not have a critical block. Then, $Right(B) - R(B) > 0$ for each $B\in \mathcal{B}$. Let $d_{min}  = \min_{B\in \mathcal{B}}\left(\frac{Right(B) - R(B)}{|B|}\right)$ and shrink each lid by $d_{min}$, fixing the left endpoint of each maximal block the lid is in, and alining the shrunken lids so that that they still form a block. Observe that after shrinking each lid, each block $B\in \mathcal{B}$ decreases by  $d_{min}|B|\leq (Right(B) - R(B))$. Therefore, each shrunken block $B$  still covers the segments that it originally covered in $H^*$ affecting only its right side. Thus, any high priority point originally covered by two blocks will still be covered by two blocks. Consider any high priority point $p$ that was originally doubly covered by a single block $B = (\ell_1,\ell_2,\dotsc,\ell_j)$ that is the right endpoint of two lids $\ell_i$ and $\ell_{i+1}$ in that block.  Suppose then that $p$ is not an endpoint of a high priority segment. Therefore, there is another block $B'$ that covers $p$ since $p-\epsilon$ and $p+\epsilon$ is single covered by $\ell_i$ and $\ell_{i+1}$, respectively. Suppose now that $p$ is an endpoint of a high priority segment. Therefore, the maximal block $B$ has a sub-block $(\ell_1,\ell_2,\dotsc,\ell_i)$ that is a critical block, contradicting our assumption that there are no critical blocks.

To see that the every low priority point is still covered by at least one lid after the lid reduction, suppose $p$ is a low priority point, and suppose $B$ is a maximal block containing $p$ in the original covering. If $B$ contains an element of $H^*$ greater than $p$, then, since we already saw that the shrinking of the lids does not uncover priority points, it will not uncover $p$. If $B$ contains no element of $H^*$ greater than $p$, let $\ell_j$ be the last lid in  $B$ containing $p$ involved in the construction of $\ESEll{k}$. By construction, if $Left(\ell_{j+1})>p$, then $Left(\ell_{j+1})=Right(\ell_j)$ since $\ell_j$ does not contain any point in $H^*$ greater than $p$ which contradicts the supposition that $B$ is maximal. 
\end{proof}

Observe that $Left(B)$ is less than $1$. In the following lemma, we show that the blocks in  $\ESEll{k}$ are equivalent when we replace  $H$ by $H^*$.
 
\begin{lemma}\label{lem:blocks}
Let $B$ be a maximal block in  $\ESEll{k}$. Then, $Left(B) = L^*(B)$ and $Left(B)$ is the left endpoint of a segment of $H^*$. 
\end{lemma}

 \begin{proof} (Lemma~\ref{lem:blocks})
Let $\ell$ be the leftmost lid in $B$. Then in the construction of $\ESEll{k}$, $\ell$ could be either 
\begin{itemize}
\item $\ell_1$, in which case $Left(B)$ is the left endpoint $0$ of a segment of $H^*$,
\item $\ell_2$, in which case $Left(B)$ is $L(\ell_1)$, which is either the left endpoint of a segment in $H$ (hence $H^*$) 
or is the right endpoint of $\ell_1$ in $\ESEll{k}$ (contradicting the maximality of $B$),
\item $\ell_j$ (for $j\ge 3$), in which is divided into two cases:
\begin{itemize}
\item $\ell_{j-1}\neq \ell_{j-2}$, in which case $Left(B) = Left(\ell_j) = L(\ell_{j-1}\setminus \ell_{j-2})$, which is either the left endpoint of a region in $H$ or the right endpoint of either $\ell_{j-1}$ or $\ell_{j-2}$ (contradicting the maximality of $B$).
\item $\ell_{j-1}=\ell_{j-2}$, in which case $Left(B)$ is $\min\{1,Right(\ell_{j-1})\}$. However, since $Left(B)<1$, $Left(B) = Right(\ell_{j-1})$. But then $Left(B)$ is the right-endpoint of another lid in $\ESEll{k}$ (contradicting the maximality of $B$).
\end{itemize}
\end{itemize}
In each of the non-contradictory cases, $Left(B)$ is the left endpoint of a segment of $H^*$, and hence $Left(B) = L^*(B)$. 
\end{proof}

 Similarly, we can show that $\ESErr{k}$  has a critical block. In the following lemma, we show an essential relation within the left and right shifted double lid cover. 

\begin{lemma}\label{lemma:smaller}
For every $\elll_{i} \in \ESEll{k}$ and $\ellrr{i} \in \ESErr{k}$, $Left(\elll_{i}) \ge Left(\ellrr{i})$.
\end{lemma}

\begin{proof} 
Given $\ESEll{k}$, we will construct $\ESErr{k}$ with the desired properties. First consider $\elll_{k} \in \ESEll{k}$. Let $\ellrr{k}  =  [1 - \Lambda_{k}, 1]$. It is easy to see that $Left(\elll_{k}) \ge Left(\ellrr{k})$. Moreover, the set $$\{\elll_{1}, \elll_{2}, ... ,\elll_{k-1},\ellrr{k} \}$$ strong double covers $C$. Inductively, suppose that the set $$\{\elll_{1}, \elll_{2}, ...,\elll_{j},\ellrr{j+1},...,\ellrr{k} \}$$ strong double covers $C$ such that $Left(\elll_{i}) \geq Left(\ellrr{i})$ for all $i > j$. Let $\ellrr{j} = [R(\elll_{j}) - \Lambda_{k}, R(\elll_{j})]$. Since $R(\elll_{j}) \geq  Left(\elll_{j+1}) \geq Left(\ellrr{j+1})$, the set $\{\elll_{1}, \elll_{2}, ...,\elll_{j-1},\ellrr{j},...,\ellrr{k} \}$ still strong double covers $C$ and $Left(\elll_{j}) \geq Left(\ellrr{j})$. Observe that  $\ESErr{k}$ consists of $k$ lids since otherwise  $\ESEll{k}$ would not be optimal. The lemma follows.  
\end{proof}

\section{Lower-Bound}~\label{sec:lower}
We use the concept of single lid cover and double lid cover to derive the lower bound when $k$ robots are being used. Recall that $\lambda_{k-1}$ is the minimum lid length that accepts a single lid cover with $k-1$ lids and $\Lambda_{2k}$ is the minimum lid length that accepts a strong double lid cover with $2k$ lids. We compare $\lambda_{k-1}$ and $\Lambda_{2k}$  and use the properties to show the lower bound. More specifically, we show that the lower bound for the patrolling problem with  $k$ robots is at least  $2\min(\lambda_{k-1}, \Lambda_{2k})$. 

We define the indicator function that determines if two consecutive lids cover at least one high priority point as follows:
$$\mathbb{I}_i = 
\begin{cases}
1 & \mbox{ if } \ell_i \cap \ell_{i+1} \cap H   \not= \emptyset \\ 
0 & \mbox{ otherwise }
\end{cases}
$$
A \emph{component} of a double cover is a set $W=\{\ell_i, \ell_{i+1}, .., \ell_{w}\}$ of consecutive lids of the cover such that  
$\mathbb{I}_i, \mathbb{I}_{i+1}, .., \mathbb{I}_{w-1}$ are all one.  We say that $W$ is a \emph{maximal component} if either $\mathbb{I}_{i-1} = 0$ or $\mathbb{I}=1$  and either $\mathbb{I}_{w} = 0$ or $w=2k$. For each component $W$, let $f(W)$ be the minimum number of lids of length $\Lambda_k$ needed to singly cover the priority segments in $W$. The following lemma is useful for determining when $\lambda_{k-1} > \Lambda_{2k}$

\begin{lemma} \label{lem:f}
Let $W$ be a component of $\mathcal{S}_{2k}$. Then $f(W) = \lfloor |W|/ 2 \rfloor.$
\end{lemma}

\begin{proof}
We can remove every other lid, starting with the leftmost lid for the first overlap of $W$ and still single cover the region in the overlaps of $W$. When $|W| = 2a$ (is even), the number of lids needed to cover the high priority points is $a=|W|/2$. When $|W|=2a+1$ (is odd), the number of lids needed to cover the high priority points is $a=(|W|-1)/2$. 
\end{proof}

An essential property of strong double lid covers when $\Lambda_{2k} <  \lambda_{k-1}$ is that the low priority segments must be of a short length and uniformly distributed.  In other words, every lid must cover at least one high priority point as we show in the following lemma. 

\begin{lemma}\label{lem:all}
If $\Lambda_{2k} <  \lambda_{k-1}$, then every lid of  $\mathcal{S}_{2k}$ covers at least one point of $H$.
\end{lemma}

\begin{proof}
Given a double cover $\mathcal{S}_{2k}$, the set of even lids cover all the high priority points as well as the set of odd lids. Assume that $\ell_i$ is not covering any high priority point. Therefore, if $i$ is even, $k-1$ even lids cover all the high priority points and therefore, $\Lambda_{2k} \geq  \lambda_{k-1}$ contradicts the assumption. Similarly, if $i$ is odd. 
\end{proof}

The \emph{component partition} of $\mathcal{S}_{2k}$  is the set $\{W_1,W_2,\dotsc,W_m\}$ of all maximal components ordered from left to right. Next we show that if $\Lambda_{2k} \leq  \lambda_{k-1}$, the size of each maximal component $|W_i|$ is even.

\begin{lemma}\label{lem:max}
Let $\{W_{1}, W_{2}, \dotsc, W_{m} \}$ be the component partition of $\mathcal{S}_{2k}$. Then, $|W_i|$ is even for all $i \in [1,m]$ provided that $\Lambda_{2k} <  \lambda_{k-1}$.
\end{lemma}

\begin{proof}
Consider the indicative function 
$$\mathfrak{I}_W = 
\begin{cases}
1 &\mbox{ if } |W| \mbox{ is odd }\\
0 &\mbox{ otherwise }
\end{cases}
$$
Let $w_{odd} = \sum_{i=1}^m \mathfrak{I}_{W_i}$. Observe that $\sum_{i=1}^m f(W_i) + w_{odd} \leq k$ since at most $k$ lids cover $H$. Therefore, if $w_{odd} > 0$ we can cover the priority segments with less than $k$ lids of length $\Lambda_{2k}$. 
\end{proof}

The lower-bound proof is based on finding $k+1$ points far apart. Thus, $k$ robots are forced to move constantly to visit $k+1$ points.  The following theorem provides such a condition.

\begin{theorem}\label{generalLowerBound}
Suppose $C$ contains $p_1,p_2, .\dotsc, p_{k+1}$ points  where $k$ of the points are high priority and $p_{j+1}-p_j \ge x$ for each $j \in [1,k]$,  then, $I^*_k \geq 2x$ where  $x$ is a positive real number.
\end{theorem}

\begin{proof}
Suppose toward a contradiction that there is a strategy with idle time less than $2x$. Let $p_i$ be the  point of low priority. Let $t_i$ be the time that one robot, say $r_j$, visits $p_i$. Since every point in $C$ is visited infinitely often such a time exists; see Figure~\ref{fig:kpoints}. Observe that at time $t_i$, $p_i$ splits the number of points and robots into  $i-1$ points and $j-1$ robots to the left and $k - i$ points and $k-j$ robots to the right. Either there are more priority points than robots to the right of $p_i$ (i.e.\ $j > i$), or there are at least as many priority points as robots to the left of $p_i$ (i.e.\ $j\le i$). We consider the case where $j> i$ (the case where $j\le i$ is symmetric). Since $p_i$ is at distance at least $x$ from $p_{i+1}$, the last time that $r_j$ could have visited $p_{i+1}$  was not later than $t_i - x$ and the earliest time that $r_j$ can visit $p_{i+1}$ is $t_i + x$. Therefore, another robot, say $r_{j+1}$, must visit $p_{i+1}$ at time  $ t_{i+1} \in [t_i - x, t_i + x]$. Observe that $j$ robots are on the left of $p_{i+1}$ and they cannot visit $p_{i+1}$ at time  $ t_{i+1} \in [t_i - x, t_i + x]$ as well. Similarly, $r_{j+1}$ cannot visit $p_{i+2}$ in the interval $[t_{i+1} - x, t_{i+1} + x]$. Therefore, another robot, say $r_{j+2}$, must visit $p_{i+2}$ at time $ t_{i+2} \in [t_{i+1} - x, t_{i+1} + x]$. Continuing this reasoning, the last robot, say $r_k$ cannot visit $p_{k+1}$ in the interval $[t_{k+1} - x, t_{k+1} + x]$  leaving $p_{k+1}$ unvisited for at least $2x$ time which contradicts the assumption that the idle time is less than $2x$. In the case where $j\le i$, we consider $r_j$ moving to the left of $p_i$ and symmetrically run into the same contradiction.
\begin{figure}[htbp]
\centering
\includegraphics[scale=1.0]{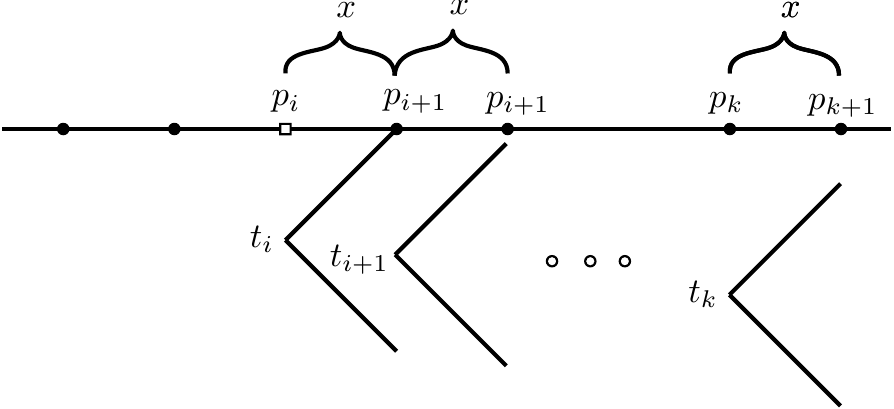}
\caption{The  minimum idle time for  $k$ robots with $k+1$  points with $k$ points being high priority at distance at least $x$ is at least $2x$.}
\label{fig:kpoints}
\end{figure}
\end{proof}

Let $\{W_{1}, W_{2}, \dotsc, W_{m} \}$ be the component partition of $\mathcal{S}_{2k}$.  From Lemma~\ref{lem:all}, every lid covers at least one high priority point if  $\Lambda_{2k} <  \lambda_{k-1}$. Therefore, for any two consecutive maximal consecutive lids, $W_{i}, W_{i+1}$  there exist $\ell \in W_{i}$ and $\ell' \in W_{i+1}$ such that $\ell \cap \ell' \cap H = \emptyset$.

We say that a set of points $P$ of $C$ is in \emph{general position} if for all $a,b,c,d\in P$, with $a\le b$ and $c\le d$, if $b-a$ is a rational multiple of $d-c$, then $a=c$ and $b=d$. We say that a high priority set $H$ is in \emph{general position} if the set of all endpoints of intervals in $H$ are in general position.

\begin{lemma}\label{lem:dcritical1}
Suppose that $\Lambda_{2k} < \lambda_{k-1}$ If $H$ is in general position, then $\ESEll{2k} \cap \ESErr{2k} = \{\elll_{a(1)},\elll_{a(2)}, ...,\elll_{a(b)}\} = \{\ellrr{a(1)},\ellrr{a(2)}, ..., \ellrr{a(b)}\}$ is the only critical block.
\end{lemma}

\begin{proof}
From Lemma~\ref{lem:dcritical}, $\ESEll{2k}$ and $\ESErr{2k}$ each has a critical block. The length of a critical block must be a multiple of the lid-length $\Lambda_{2k}$, and hence any two critical blocks (of potentially distinct double  $\Lambda_{2k}$-lid covers) would have lengths that are rational multiples of each other, violating the assumption that the priority set $H$ is in general position. Thus the critical block is unique, and must be the same for any two double $\Lambda_{2k}$-lid  covers. Therefore, $\{\elll_{a(1)},\elll_{a(2)}, ...,\elll_{a(b)}\} = \{\ellrr{a'(1)},\ellrr{a'(2)}, ..., \ellrr{a'(b)}\}$.

We need to show that $\elll_{a(i)} = \ellrr{a'(i)}$ for all $i$. Suppose by contradiction that $\elll_{a(i)} \not= \ellrr{a'(i)}$ for some $i$. From Lemma~\ref{lemma:smaller}, $Left(\elll_j) \geq Left(\ellrr{j})$ for all $j$. Therefore, $a'(i) <  a(i)$. Suppose that $a'(i) +1 =  a(i)$. Therefore,  $Left(\elll_{a(i)})\geq Left(\ellrr{a'(i)+1}) \geq Left(\ellrr{a'(i)})$ and  the set $\{\elll_{1}, \elll_{2}, .., \elll_{a(i)-1}, \elll_{a(i)} =  \ellrr{a'(i)+1}, ... \ellrr{2k}\}$ strong double covers $C$ with $2k-1$ lids. Therefore, from Lemma~\ref{lem:all}, $\Lambda_{2k} \geq \lambda_{k-1}$ which contradicts the assumption. 
\end{proof}

\begin{theorem}\label{thm:opt1}
If the priority set $H$ is in general position, then $I^*_k \geq  2\min(\Lambda_{2k}, \lambda_{k-1})$.
\end{theorem}

\begin{proof}
Let $\ESEll{2k} = \{\elll_1, \elll_2, \dotsc, \elll_{2k}\}$  and $\ESErr{2k}  = \{\ellrr{1}, \ellrr{2}, \dotsc, \ellrr{2k}\}$ be a strong shift left and right double $\Lambda_{2k}$-lid covers of $C$ with $2k$ lids.Let $\mathcal{W}_{k-1}^\rightarrow = \{\elll_1, \elll_2, ..., \elll_{k-1}\}$ and $\mathcal{W}_{k-1}^\leftarrow = \{\ellrr{1}, \ellrr{2}, ..., \ellrr{k-1}\}$ be the lids of the left shift  and right shift  $\lambda_{k-1}$-lid covers of $H$ with $k-1$ lids.

First consider the case where $\Lambda_{2k} <  \lambda_{k-1}$. Let $\{W_{1}^\rightarrow, W_{2}^\rightarrow, \dotsc, W_{m}^\rightarrow \}$ be the component partition of $\ESEll{2k}$ and $\{W_{1}^\leftarrow, W_{2}^\leftarrow, \dotsc, W_{m}^\leftarrow \}$ be the component partition of $\ESErr{2k}$. From Lemma~\ref{lem:max}, $W^\leftarrow_{i}$ and  $W^\rightarrow_{i}$ have even number and from Lemma~\ref{lem:all}, every lid covers at least one point in  $H$.

Recall that from Lemma~\ref{lem:dcritical}, $\ESEll{2k}$ and $\ESErr{2k}$ have a critical block each. Further, since the high priority segments are in general position, the critical blocks are unique from Lemma~\ref{lem:dcritical1}. Let $B = \{\ell_{a(1)}, \ell_{a(2)}, \dotsc, \ell_{a(b)}\}$ be the critical block (of both $\ESEll{2k}$ and $\ESErr{2k}$). The proof is based on three claims. In the first claim we show how to select the points in the segment determined by $B$, meanwhile in the second and third claims we show how to select the points on the left and right of $B$ such that all points are at a distance at least $\Lambda_{2k}$ from each other and $k$ of them are high priority.

\begin{claim}  
Let $P_B =  \{ p : p = left(\ell_{a(l)}) \mbox{  or } p=Right(\ell_{a(l)})  \forall \ell_{a(l)} \in B\}$. Every two  points in $P_B$ are at distance at least $\Lambda_{2k}$, $|P_B| = 1 + b $ and if
\begin{compactenum}
\item $i = j$ and $a(1)$ and $a(b)$ are even  then all the points in $P_B$  are in $H \cup \{1\}$. 
\item $i = j$ and $a(1)$ and $a(b)$ are odd  then all the points in $P_B$  are in $H \cup \{0\}$. 
\item $i \not= j$ then $j = i + 1$,  $a(1)$ is even, $a(b)$ is odd and  $P_B$ has $b$  high priority points.
\end{compactenum}
\end{claim}

\begin{proof}  (Claim 1.)
We use the main observation that the common point $p$ of every two consecutive lids $\ell_{a(l)}$ and $\ell_{a(l')}$ in a block is either high priority which implies that there is a high priority segment that starts on the left of $p$ and finishes on the right of $p$, or a low priority point which implies that $\ell_{a(l)}$ and $\ell_{a(l')}$ belong to two different components. In the former case, the high priority segment must be covered by a lid not in the block. Observe that since $B$ is a critical block, $Left(\ell_{a(1)})$  and $Right(\ell_{a(b)})$ are in $H^*$.

From the construction of $P_B$, it is not difficult to see that every two points in $P_B$ are at distance at least $\Lambda_{2k}$ and $|P_B| = 1 + b $. We show each case separately. 
\begin{compactenum}
\item Since all the lids in the block are maximal and in the same component, there must be a lid not in $B$ in-between every two consecutive lids in the block. That implies that $a(1)$ and $a(l)$ are even for all lids in $B$. Therefore, $Left(B)$ is high priority since
from Lemma~\ref{lem:f}, each component $W^\rightarrow_{i}$ has even number of lids. Further, $Right(B) \in H \cup \{1\}$ since $B$ is a critical block.

\item Similar as before, since all the lids in the block are in the same component, there must be a lid not in $B$ in-between every two consecutive lids in the block. That implies that $a(1)$ and $a(l)$ and are odd for all lids in $B$. Therefore, $Right(B)$ is high priority since from Lemma~\ref{lem:f}, each component $W^\rightarrow_{i}$ has even number of lids. Further, $Left(B) \in H \cup \{0\}$ since $B$ is a critical block.

 \begin{figure}[htbp]
\centering
\includegraphics[scale=0.71]{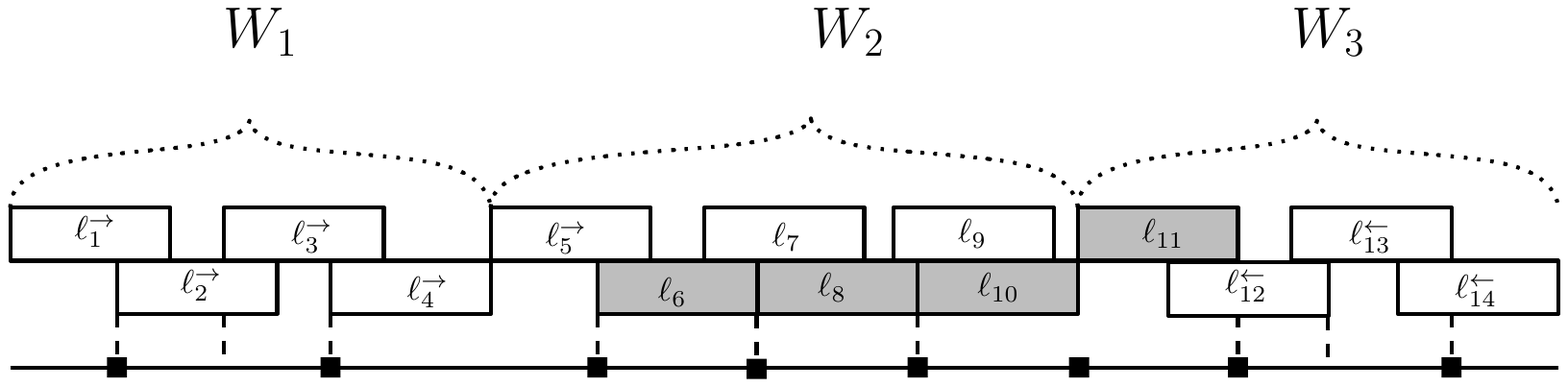}
\caption{Block $B=\{\ell_6, \ell_8, \ell_{10}, \ell_{11}\}$. (Vertical dashed lines represent high priority point and black square the chosen of points.) }
\label{fig:kpoints1}
\end{figure}

\item Since $i \not= j$, there must exist a lid $\ell_{a(l)}$ that is the right end lid of the component $W_i$. Therefore, $a(l)$ is even since from Lemma~\ref{lem:f}, $W^\rightarrow_{i}$ has even number of lids; Refer to Figure~\ref{fig:kpoints1}. In consequence $a(1)$ is also even since both are in $W^\rightarrow_{i}$. Suppose that $j > i +1$. Therefore, the last lid in the block in the component  $W^\rightarrow_{i+1}$ is odd since $a(l)+1$ is odd which contradicts that  every component is even. Therefore, $j= i+1$. It is not difficult to see that $Right(\ell_{a(l)}) = Left(\ell_{a(l)+1})$ is the unique low priority point.
\end{compactenum}
\end{proof}
 
\begin{claim} 
We can select a set $P_B^-$ that consists of $\lfloor \frac{a(1)-1}{2}\rfloor $ points at distance at least $\Lambda_{2k}$ from each other and from every point in $P_B$ such that  if $a(1)$ is even, $|P_B^-|$  points are  high priority, otherwise  $|P_B^-|-1$ points are high priority.
\end{claim}

\begin{proof} (Claim 2) We prove each case independently:
\begin{compactitem}
\item $a(1)$ is even:  Let $P_B^- = \{ p : Left(\elll_{2l}) \forall l \in [1, \frac{a(1)-2}{2}]\} $. Clearly, $|P_B^-| = \frac{a(1) -2}{2}= \lfloor \frac{a(1)-1}{2}\rfloor $. Since each component is even,   $Left(\elll_{2l})$ covers a high priority point for each lid $\elll_{2l}$. Therefore, $|P_B^-|$ points in $P_B^-$ are  high priority. Further, they are at distance at least $\Lambda_{2k}$ from each other and from every point in $P_B$; see Figure~\ref{fig:kpoints2}.
 \begin{figure}[htbp]
\centering
\includegraphics[scale=0.7]{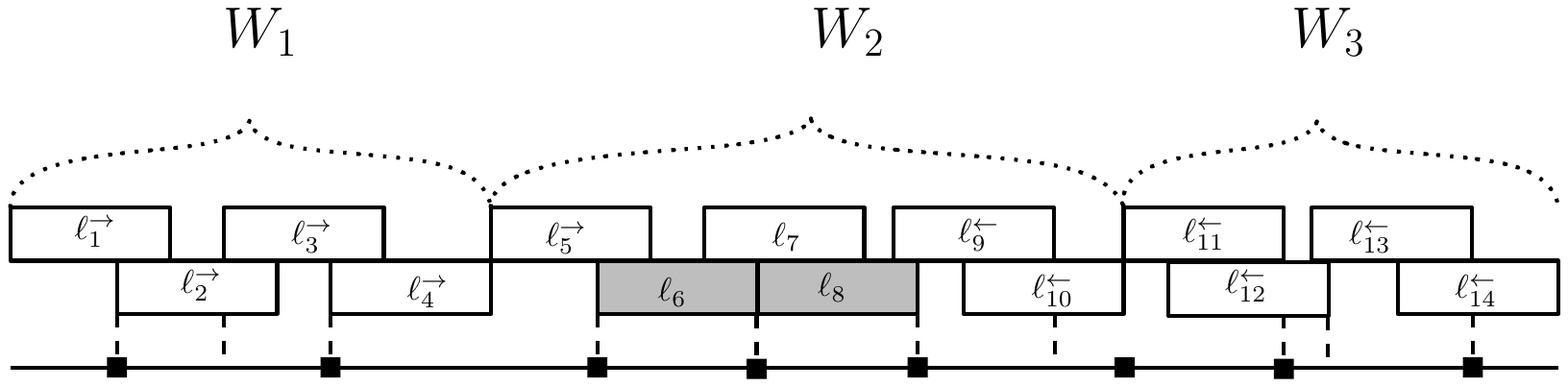}
\caption{Block $B=\{\ell_6, \ell_8\}$. (Vertical dashed lines represent high priority point and black square the chosen of points.) }
\label{fig:kpoints2}
\end{figure}

\item $a(1)$ is odd:   
Suppose that $\elll_c$ is the leftmost lid in $W_i$. Since each component has even number of lids,  $c$ is odd.  Let 
$$ P_B^- =  \left\{p : 
\begin{cases}
Left(\elll_{2l}) & \forall l \in [1, \frac{c-1}{2}]  \mbox{ \bf or } \\
 Left(\elll_{2l-1}) & \forall l \in [\frac{c+1}{2}, \frac{a(1)-1}{2}] 
\end{cases}
 \right\}$$
 Thus, the number of points is $(\frac{a(1)-1}{2} -\frac{c+1}{2} +1)  + (\frac{c-1}{2} - 1 +1)= \frac{a(1)-1}{2}$.  Observe that  only $Left(\elll_{c})$ is low priority since $c$ is the leftmost lid in $W_i$. Therefore,   $|P_B^-| -1$ points in $P_B^-$ are high priority;
 see Figure~\ref{fig:kpoints3}.
  \begin{figure}[htbp]
\centering
\includegraphics[scale=0.7]{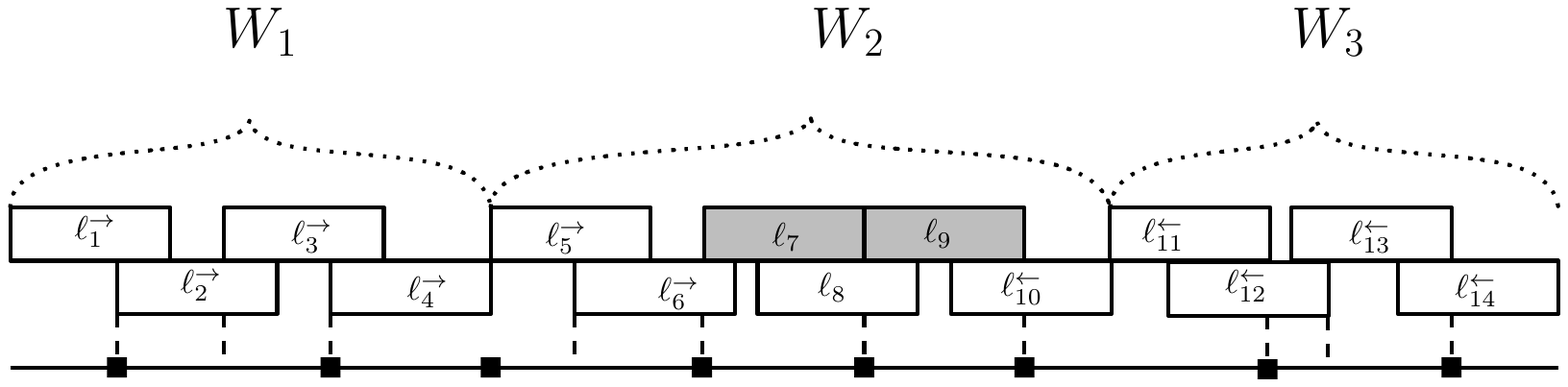}
\caption{Block $B=\{\ell_7, \ell_9\}$. (Vertical dashed lines represent high priority point and 
black square the chosen of points.) }
\label{fig:kpoints3}
\end{figure}
\end{compactitem}
\end{proof}

\begin{claim}  
We can select  a set $P_B^+$ that consists of  $\lfloor \frac{2k - a(b)}{2} \rfloor$ points at distance at least $\Lambda_{2k}$ from each other and from every point in $P_B$ such that if $a(b)$ is even, $|P_B^-|-1$  points are  high priority, otherwise  $|P_B^-|$ points are high priority.
\end{claim} 

\begin{proof} (Claim 3) We prove each case independently:
\begin{compactitem}
\item $a(b)$ is even:  
Suppose that $\ellrr{c}$ is the right most lid in $W_i$. Since each set of consecutive lids is even,  $c$ is even. Let 
$$P_B^+ =  \left\{p : 
\begin{cases}
Right(\ellrr{2l-1}) & \forall l \in [\frac{c+2}{2}, k]  \mbox{ \bf or }\\
Right(\ellrr{2l}) & \forall l \in [\frac{a(b)+2}{2}, \frac{c}{2}, ] 
\end{cases}
 \right\}$$
Thus,  $|P_B^+|= (k -  \frac{c+2}{2} +1)+  (\frac{c}{2} -\frac{a(b)+2}{2} + 1) = \frac{2k - a(b)}{2} $; see Figure~\ref{fig:kpoints2}.  Observe that  only $Right(\ellrr{c})$ is low priority since $c$ is the right most lid in $W_i$. Therefore,  $|P_B^+|-1$  points in $P_B^-$ are high priority. 
\item $a(b)$ is odd:  Let $P_B^+ = \{ p : Right(\ellrr{2l-1}) \forall l \in [\frac{a(b)+3}{2}, k]\}$. Clearly, $|P_B^-| = k - \frac{a(b)-3}{2} + 1 = \frac{2k - a(b)-1}{2} $; see  Figure~\ref{fig:kpoints3}. Since $Right(\ellrr{2l})$ covers a high priority point for each lid $\elll_{2l+1}$, 
$|P_B^+|$ points in $P_B^+$ are  high priority. Further, they are at distance at least $\Lambda_{2k}$ from each other and from every point in $P_B$.
\end{compactitem}
\end{proof}

Now we show that $P = P_B^- \cup P_B \cup P_B^+$ consists of $k+1$ points at distance at least $\Lambda_{2k}$ from each other where $k$ are high prioirity.  From Claim~1,~2~and~3,  $|P| = (\lfloor \frac{a(1)-1}{2}\rfloor)  + (\lceil \frac{a(b) - a(1)}{2} \rceil + 2) + (\lfloor \frac{2k - a(b)}{2} \rfloor)$. We consider three cases:

\begin{compactitem}
\item $a(1)$  and $a(b)$ are odd: 
$$
\begin{array}{rl}
|P| & = \frac{a(1)-1  + a(b) - a(1) + 4 + 2k - a(b)-1}{2} \\
    & = \frac{2k + 2}{2} \\
    &  =  k+1. 
\end{array}
$$
And the number of critical points are  $$|P_B^-|-1 + |P_B|  + |P_B^+| = k.$$

\item $a(1)$  and $a(b)$ are even: 
$$
\begin{array}{rl}
|P| &= \frac{a(1)-2  + a(b) - a(1) + 4 + 2k - a(b)}{2}  \\
&=  \frac{2k +2}{2}  \\
&= k+1.
\end{array}
$$ 
And the number of critical points are $$|P_B^-| + |P_B|  + |P_B^+|-1 = k.$$
\item $a(1)$  is even and $a(b)$ is odd: 
$$
\begin{array}{rl}
|P| &= \frac{a(1)-2  + a(b) - a(1) +1 + 4 + 2k - a(b)-1}{2}\\
& =  \frac{2k+2}{2} \\
& = k+1.
\end{array}
$$
And the number of critical points are $$|P_B^-| + |P_B| -1 + |P_B^+| = k.$$
\end{compactitem}

The theorem follows when  $\Lambda_{2k}<  \lambda_{k-1}$ since we can select $k+1$ points at distance a least $\Lambda_{2k}$ such that $k$ are high priority.

Now we consider $\Lambda_{2k} \geq  \lambda_{k-1}$. First, we show in the next claim that there exists a common critical block in the right and left shifted single $\lambda_{k-1}$-lid cover.

\begin{claim} 
$\mathcal{W}_{k-1}^\rightarrow$ and $\mathcal{W}_{k-1}^\leftarrow$ has a common maximal critical block.
\end{claim}

\begin{proof} (Claim 4)
Let $B^\rightarrow = \{\elll_{a(1)},  ..., \elll_{a(b)}\}$ be the critical block in $\mathcal{W}_{k-1}^\rightarrow$ and let $B^\leftarrow = \{\ellrr{a'(1)}, ..., \ellrr{a'(b')}\}$ be the critical block in $\mathcal{W}_{k-1}^\leftarrow$. Suppose that $B^\rightarrow \not=  B^\leftarrow$. Thefore, there must exist a lid $\ellrr{i}$ such that $Left(\ellrr{i}) \leq Right(\elll_{a(b)}) < Right(\ellrr{i})$, otherwise $B^\rightarrow =  B^\leftarrow$ since $Right(\elll_{a(b)}) = Right(\ellrr{a'(b')})$. Observe that $Left(\elll_{j}) \geq Left(\ellrr{j})$. Therefore, $\lambda_{k-1}$ is not minimum since $$\{\elll_{1}, ...,\elll_{a(1)}, ..., \elll_{a(b-1)}, \ellrr{i}, ..., \ellrr{k-1}\}$$ has not critical block. 
\end{proof}

We prove by contradiction. Suppose that there are not $k+1$ points at distance at least $\lambda_{k-1}$ from each other where $k$ points are high priority. Let $\{\elll_{a(1)}, ..., \elll_{a(b)}\} =\{\ellrr{a(1)},  ..., \ellrr{a(b)}\}$ be the critical block. First, we consider the set $P$ of $k$ high priority points as follows:
$$
 p_j = \begin{cases}
Left(\elll_{j})  & \text{if } j \in [1, a(b)]\\
Right(\ellrr{j})  & \text{if } j \in [a(b), k-1]\\
 \end{cases}
 $$
 
Observe that $p_{j-1} -p_{j} \geq \lambda_{k-1}$. Therefore,  if  either $p_{j-1} -p_{j} \geq 2\lambda_{k-1}$ or $p_{1} \geq \lambda_{k-1}$ or $p_{k} \leq 1 -\lambda_{k-1}$, then we can select a point that is at distance at least $\lambda_{k-1}$ from every point in $P$. Let us assume that there does not exist any point at distance at least $\lambda_{k-1}$ from every point in $P$.

Let $\ell_1 = [0, \lambda_{k-1}]$ and $\ell_{j} = [L(\elll_{j-1} \setminus \ell_{j-1}), L(\elll_{j-1} \setminus \ell_{j-1}) + \lambda_{k-1}]$ for $1 < j \leq  a(1)$; see Figure~\ref{fig:kcontradition1}. Suppose that there exists $j$ such that $Right(\ell_{j})  \leq  Left(\elll_{j})$. Let $j < a(1)$ be the largest index such that $Left(\ell_j)$ is low priority point. Consider the set of  $k+1$  points as follows:
$$
 p'_j = \begin{cases}
Left(\elll_{l})  & \text{if } l \in [1, j-1]\\
Left(\ell_{l})  & \text{if } l \in [j, a(b)]\\
Right(\ellrr{l})  & \text{if } l \in [a(b), k-1]\\
 \end{cases}
 $$
 Observe that $p'_{j-1} -p'_{j} \geq \lambda_{k-1}$  and $k$ are high priority (only $Left(\ell_j)$) is a low priority points). Therefore, we can assume that $Right(\ell_{j})  >  Left(\elll_{j})$ for all $j$.

Let $\ell_k = [1-\lambda_{k-1}, 1]$ and $\ell_{j} = [R(\ellrr{j} \setminus \ell_{j+1}), R(\elll{j} \setminus \ell_{j+1}) + \lambda_{k-1}]$ for $a(b) \leq  j <  k$. Similarly we can show that  $Left(\ell_j)  > Right(\ellrr{j-1})$.

Thus, we can assume that  $Right(\ell_{j})  <  Left(\elll_{j})$ and  $Left(\ell_j) > Right(\ellrr{j-1})$. Let, $$\mathcal{S} = \{\{\elll_{l} | l \leq a(b)\} \cup  \{\ellrr{l} | l \geq a(1)\} \cup \{\ell_{l} | l \leq a(1)\} \cup \{\ell_{l} | l > a(b)\}  \}.$$ Observe that $S$ is a valid strong double cover with $a(b)  + (k-1 + a(1) +1) + a(1) + (k-a(b)) = 2k$ lids. Furthermore, $B$ is not longer a critical block since since $R(\elll_{l}) <  Right(\elll_{l})$ and $R(\ell_{l}) <  Right(\ell_{l})$ for $l$ less or equal to $a(1)$ as well as $L(\ellrr{l}) > Left(\ellrr{l})$  and $L(\ell_{l}) > Left(\ell_{l})$ for $l \geq a(b)$, which contradicts the fact that $\Lambda_{2k} \geq \lambda_{k-1}$. The theorem follows. 

\begin{figure}
\centering
\includegraphics[scale=0.8]{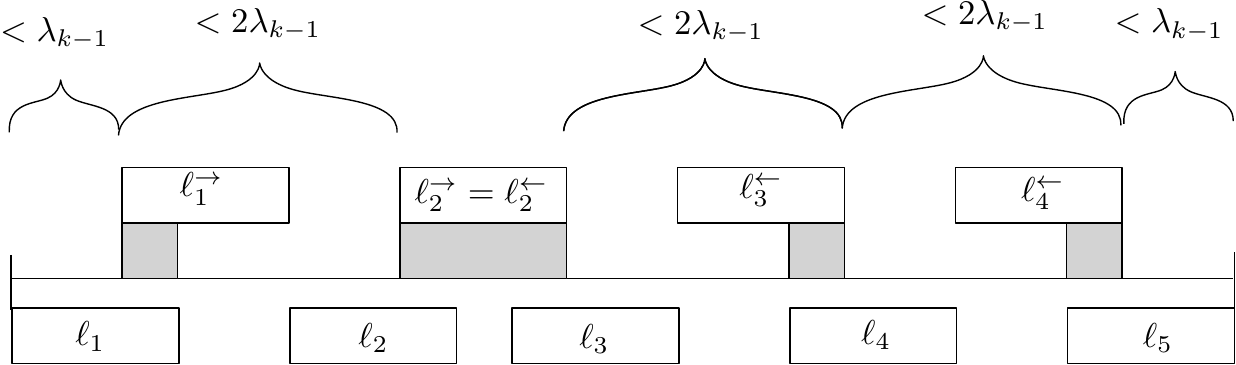}
\caption{Double strong $\lambda_{k-1}$-lid cover with $2k$ lids obtained when there is not a set of high priority points $P$ such that $p_j - p_{j-1}  < 2\lambda_{k-1}$,  $p_1  < \lambda_{k-1}$ and  $p_{k}  > 1 - \lambda_{k-1}$.}
\label{fig:kcontradition1}
\end{figure}
\end{proof}

\section{Upper bound}~\label{sec:upper}

In this section, we provide upper bounds for the idle time. Given $k$ robots, our approach considers a weak cover of $H$ with $k-1$ lids of length  $\lambda_{k-1}$ as well as a strong double cover of $H$ with $2k$ lids of length $\Lambda_{2k}$. First, we show a simple strategy that attains an idle time of $2\lambda_{k-1}$ where we split the unit line into $k-1$ segments of equal length. Then, each robot patrols a segment and one robot moves back and forth in the unit segment. The $2\lambda_{k-1}$ idle time is attained when all robots start synchronously on the leftmost point.  Second, we provide a strategy that attains a $1.5$ approximation of the optimal idle time. We split the unit line into $k$ segments of equal length. Then each robot patrols a segment. The $3\Lambda_{k}$ idle time is attained when all robots start synchronously on the leftmost point. Finally, we show a more complex strategy that attains optimal idle time, i.e.,  $2\Lambda_{2k}$ where robots are assigned to segments of different length.

\begin{strategy} 
Let $\mathcal{W}_{k-1} = \{\ell_1, \ell_2, ..., \ell_{k-1}\}$. Let robot $r_i$ move back and forth at maximum speed in lid $\ell_i$ and robot $r_k$ move back and forth on the unique line segment. Refer to Figure~\ref{fig:upper1}.
\end{strategy}

\begin{figure}[htbp]
\centering
\includegraphics[scale=0.7]{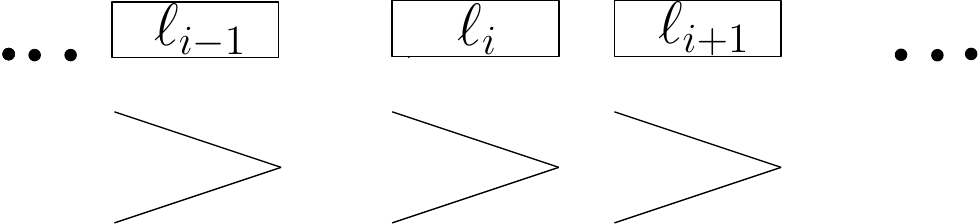}
\caption{Strategy 1.}
\label{fig:upper1}
\end{figure}

\begin{theorem}\label{thm:uppweak}
Strategy~1 attains idle time  $2\lambda_{k-1}$.
\end{theorem}

\begin{proof} 
Observe that robot $r_k$ visits all $C$ infinitely often. Further, the idle time of the high priority point in each lid is bounded by two times the length of the lid, i.e., $2\lambda_{k-1}$. 
\end{proof}

In the second strategy we consider the minimum strong double lid cover $\mathcal{S}_{2k} = \ell_1, \ell_2, ..., \ell_{2k}$. We define a cap $c_i$ as the union of two consecutive lids $\ell_{2i -1}, \ell_{2i}$, i.e., $c_{i} =  \{\ell_{2i -1}, \ell_{2i} \}$ for all $i \in [1, k]$.

\begin{strategy} 
For each $i$ let $c_i =  \{\ell_{2i -1}, \ell_{2i} \}$ be the $i$-th cap, and let $c^* = \max_i(\len(c_i))$, where $\len(c_i)$ is the length of the $i$-th cap.  For every cap $c_i$, let $x_i$ be the center of $c_i$. Place robot $r_i$ at $x_i + c^*/2$ and let it move back and forth at maximum speed in the segment $[x_i - c^*/2, x_i + c^*/2]$. Refer to Figure~\ref{fig:upper2}.
\end{strategy} 

\begin{figure}[htbp]
\centering
\includegraphics[scale=0.9]{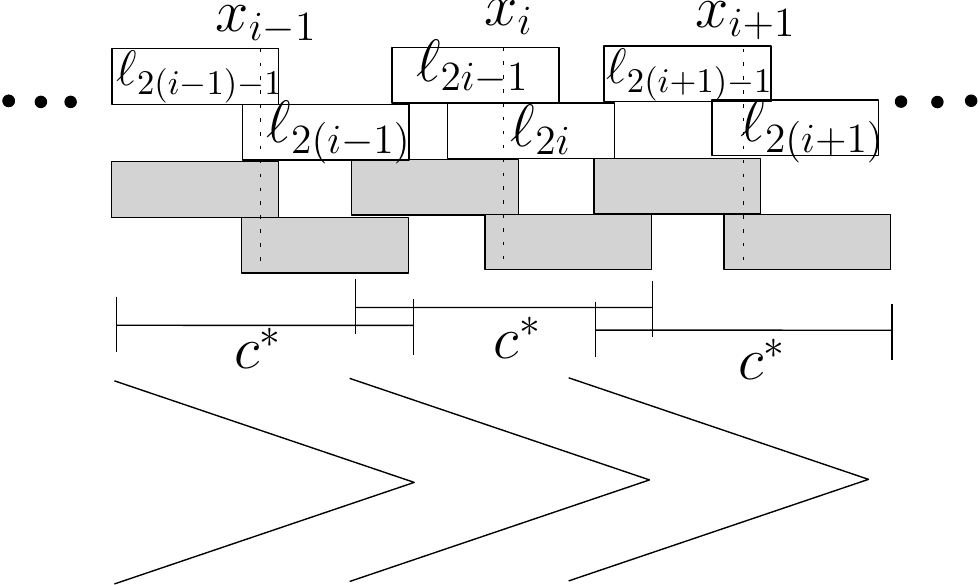}
\caption{Strategy 2.}
\label{fig:upper2}
\end{figure}

\begin{theorem}\label{thm:aprx}
Strategy~2 attains idle time of at most $3\Lambda_{2k}$. 
\end{theorem}

\begin{proof} 
It is easy to see that the strategy 2 covers the unit interval since the strong double cover already covers it. Let $c_j =  \{\ell_{2j -1}, \ell_{2j} \}$ be any cap.  First consider the high priority points in the intracap, i.e., $H \cap \ell_{2j -1} \cap \ell_{2j}$. Observe that  $r_j$ visits $Left( \ell_{2j})$ after visiting $x_j + c^*/2$.Therefore, the maximum idle time of $Left( \ell_{2j})$ is: $$2(x_j + c^*/2 - Left(\ell_{2j})) \leq 2(x_j + c^*/2 - (x_j - \Lambda_{2k} / 2)) \leq 3\Lambda_{2k} $$ since $c^* \leq  2\Lambda_{2k}$. Similarly, $r_j$ visits $Right( \ell_{2j-1})$ after visiting  $x_j - c*/2$.Therefore, the maximum idle time of  $Right( \ell_{2j-1})$ is $$2(Right(\ell_{2j-1}) - (x_j -  c^*/2))  \leq  2(x_j + \Lambda_{2k} /2-  (x_j - c^*/2))) \leq 3\Lambda_{2k} .$$

Now consider the high priority points in the intercaps, in other words the intersection of two lids of different caps. Let $c_i = \{\ell_{2i -1}, \ell_{2i} \}$ and $c_{i+1} = \{\ell_{2i + 1}, \ell_{2(i+1)} \}$ be two consecutive caps. Consider any point $p$ in $H \cap c_i \cap c_{i+1} =  H \cap \ell_{2i} \cap \ell_{2i+1}$, and assume $p\geq x_j$ (otherwise $p$ would also in the intracap region $\ell_{2i-1}\cap \ell_{2i}$, and the previous case applies). Both $r_i$ and $r_{i+1}$ will visit $p$, as their paths span all $c_i$ and $c_{i+1}$ (and possibly more). Observe that after $r_i$ visits $p$, it will take at most $c^*$ time for $r_{i+1}$ to reach $p$. Similarly, after $r_{i+1}$ visits $p$, it will take at most $c^*$ time for $r_{i}$ to reach $p$. The theorem follows since $c^*\le 2\Lambda_{2k}$. 
\end{proof}

From Theorems \ref{thm:uppweak} and \ref{thm:aprx} we obtain a 1.5 approximation  to the optimal idle time.  

\begin{corollary}\label{cor:weak}
$I^* \leq \min(2\lambda_{k-1},3\Lambda_{2k}) \leq  \frac{3\min(\lambda_{k-1},\Lambda_{2k})}{2} $.
\end{corollary}

\begin{proof}
The corollary follows from Theorems~\ref{thm:aprx}~and~\ref{thm:uppweak} 
\end{proof}

We improve the result of Corollary~\ref{cor:weak} with a more complex strategy where we also consider the minimum strong double $\Lambda_{2k}$-lid cover $\mathcal{S}_{2k} = \ell_1, \ell_2, ..., \ell_{2k}$. In the third strategy, robots do not move synchronously. Instead, robots cover all the odd lids, and eventually, they switch to the even lids one by one and return to the odd lids eventually. That process is executed infinitely often.

We say that a lid with endpoints $p$ and $q$ (with either $p<q$ or $q<p$) is \emph{periodically covered} by a robot $r$ if $r$ is returning to $p$ having just traveled from $p$ to $q$.

\begin{strategy}  
Each robot first covers (travels back and forth) the odd lids until (at time $t_0$) for each $i$ the lid $\ell_{2i-1}$ is periodically covered by $r_i$. Then, the first time after $t_0$ that $r_1$ turns right, $r_1$ continues right to cover $\ell_2$. Once $\ell_2$ is periodically covered and $r_2$ turns right, $r_2$ will continue to cover $\ell_4$, and so on until all the even lids are periodically covered. Then the process is reversed, with $r_k$ switching to covering the odd lid $\ell_{2k-1}$, and so on until finally $r_1$ switches to cover $\ell_1$, and we repeat. Refer to Figure~\ref{fig:upper3}.
\end{strategy}

\begin{figure}[htbp]
\centering
\includegraphics[scale=0.85]{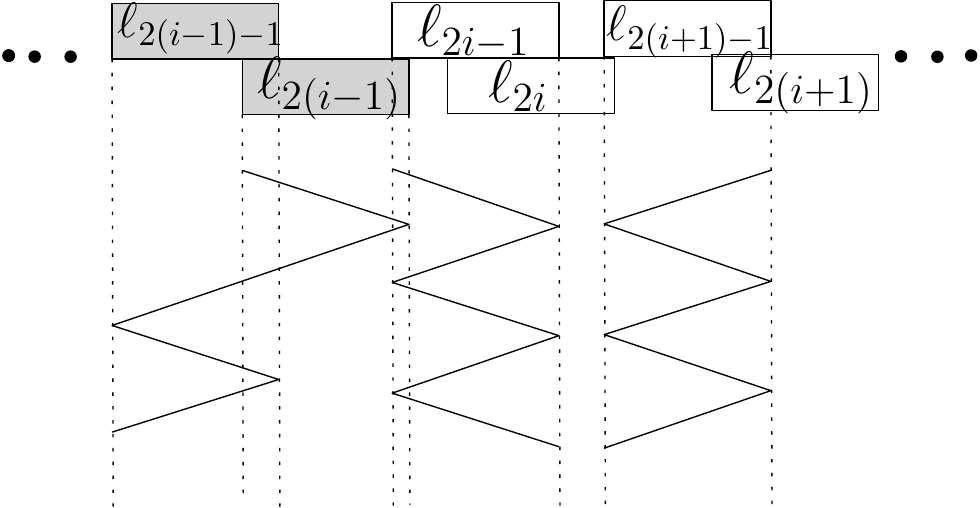}
\caption{Strategy 3.}
\label{fig:upper3}
\end{figure}

\begin{theorem}\label{thm:upper}
Strategy~3 attains idle time  $2\Lambda_{2k}$.
\end{theorem}

\begin{proof} 
The first observation is that when all the robots are either covering the even or odd lids, all high priority points are covered since each set of odd and even lids form a single $\Lambda_{2k}$-lid cover. Indeed, if a high priority point $p$ is covered by only one set, say the odd lids, then two odd lids must cover $p$ which would contradict the order of the lids. Moreover, the idle time is bounded by $2\Lambda_{2k}$. We show by induction that the idle time of the high priority points does not increase when the robots switch from odd to even. The case when the robots switch from even to odd is analogous and left to the reader as an exercise. The inductive hypothesis is that the first $i$ robots on the left are covering the even lids, the $k-i-1$ are covering the odd lids, and the idle time is bounded by $2\Lambda_{2k}$. For the base case,  consider $i=1$. Since  $\ell_{1}$ does not intersect a lid on the left,  when robot $r_i$ switches to  $\ell_{2}$, it can return to $Left(\ell_{2})$ in time $2\Lambda_{2k}$. Assume that the inductive hypothesis is true for some $i-1$ greater than one.  Consider robot $r_i$. Observe that the high priority points in $\ell_{2(i-1)} \cap \ell_{2i-1} $ are being covered by robot $r_{i-1}$ with idle time bounded by $2\Lambda_{2k}$. When robot $r_i$ switches to lid $\ell_{2i}$ it returns to $Left(\ell_{2i})$ in time $2\Lambda_{2k}$. Since the union of the odd and even lids covers $C$, the low priority points are being visited infinitely often. The theorem follows. 
\end{proof}

From Theorems \ref{thm:uppweak} and \ref{thm:upper} we obtain an algorithm that attains optimal idle time. 

\begin{corollary}
$I^* \leq 2 \min(\lambda_{k-1}, \Lambda_{2k})$.
\end{corollary}

\section{Computing Optimal Lid Covers}~\label{sec:algorithm}

In this section, we show how to efficiently compute a strong double lid cover with $2k$ lids of minimum length and a single lid cover with $k-1$ lids of minimum length. We show that they can be computed in $O(\max(k,n) \log n)$ time where $n$ is the number of high priority sections. A feasible solution is a solution where at least one block is critical. First, we show that given a length $l$, we can determine in $O(\max(k,n))$ time whether the unit segment accepts a strong double cover and a single cover with $2k$ and $k-1$ lids of length $l$, respectively. Then we obtain the lid length of a feasible solution to find the optimal value using a binary search. 

We store the high priority sections in an array $H$ of dimension $n$ where $H_i =\{left, right\}$, i.e., the initial point and final point of the $i$-high priority section. The set of lids is store in the array $L$ with the appropriate  dimension  where $L_i = \{left, right, h\}$ such that $L_i.left$ and $L_i.right$ indicates the initial and final position of the $i$-lid and $L_i.h$ is the index of the right most high priority section such that $H_{L_i.h}.left \leq right$.

Given $k$ and the lid length $l$, we can decide whether $C$ accepts a  strong double $l$-lid cover with $k$ lids in time $O(\max(k, n))$.

\begin{lemma}\label{lem:alg1}
There exists an algorithm that decides if $C$ accepts a right shifted strong double $l$-lid cover with $k$ lids in time $O(\max(k, n))$ where $n$ is the number of high priority sections.
\end{lemma}

\begin{proof} 
For the proof, we consider  two sentinel high priority sections $H_0$ and  $H_{n+1}$ where $H_0.left = H_0.right = 0$ and $H_{n+1}.left = H_{n+1}.right = 1$. We also consider one  sentinel lid $L_0$ where $L_0.h = 0$. 

We prove by induction on $i$. The inductive hypothesis is that for every $i \in [1,2k]$,   the segments $[0, L_{i-1}.right]$, $[0, L{i}.right]$  are strong double covered and single covered with $i$ lids (not necessarily disjoint), respectively and $L_i.h$ is the index of the right most high priority section such that $H_{L_i.h}.left \leq L_i.right$. Consequently, $H_{L_i.h+1}.left > L_i.right$.

Base step ($i=1$):  Let $L_i =  \{0, l, h\}$ where $h$ is the index of the high priority section such that $H_h.left  \leq L_1.right$ and $H_{h+1}.left > L_i.right$. Trivially,  the segments $[0, L_0.right]$ and $[0, L_1.right]$ are double covered and single covered, respectively.

Inductive step ($i>1$): Assume that the inductive hypothesis is true for $i-1$. Therefore, the segment $[0, L_{i-2}.right]$ is double covered and the segment $[0, L_{i-1}.right]$ is single covered. Moreover, $H_{L_{i-1}.h}.left \leq L_{i-1}.right$ and $H_{L_{i-1}.h+1}.left > L_{i-1}.right$.

\begin{algorithm}[ht!]
  \caption{Computes  right shifted strong double $l$-lid cover with $k\geq 2$ lids if $C$ 
  accepts it}
  \label{alg:acceptdouble}
  \TitleOfAlgo{DoubleLidCover}

    \SetKwInput{InOut}{Input/Output}        
        \KwIn{$H$: The set of disjoint high priority segments}
        \KwIn{$k$: The number of lids}    
    \KwIn{$l$: Lid length}            
        \KwOut{$True$:  If $C$ accepts a strong double $l$-lid cover}
      \KwOut{$L$:  The set of lids}

  $L_1 = \{0, l, h : H_h.left  \leq l$ and $H_{h+1}.left > l\}$\;
   \For{$i=2 \to k$}
   {
       \If{$H_{L_{i-2}.h}.right \leq L_{i-2}.right$}
       {  
        \If{$H_{L_{i-1}.h}.left > L_{i-2}.right$}
           {  
               $left \gets H_{L_{i-1}.h}.left$ \;
        }
        \Else{
            $left \gets L_{i-1}.right$ \;
        }       
    }
       \Else{
           $left \gets H_{L_{i-2}.h}.right$\;
    }
      $L_i = \{left, left+l, h : H_h.left  \leq left+l$ { \bf and } $H_{h+1}.left > left+l\}$\;
   }
   \Return  ($L_{k}.right  \geq 1 \mbox{ \bf and } L_{k-1}.right  \geq H_n.right$, $L$)\;
\end{algorithm}

 We consider  two cases:

\begin{compactenum}
\item 
The rightmost point of $H_{L_{i-2}.h}$ is at most $L_{i-2}.right$, i.e., $H_{L_{i-2}.h}.right \leq L_{i-2}.right$. We consider two sub cases. 
\begin{compactenum}
\item $L_{i-1} \setminus  L_{i-2}$ contains high priority points, i.e., $L_{i-2}.right < H_{L_{i-1}.h}.left \leq L_{i-1}.right$.  Let $left  = H_{L_{i-1}.h}.left$.
\item $L_{i-1} \setminus  L_{i-2}$ does not contain any one high priority point, i.e., $H_{L_{i-1}.h}.left > L_{i-1}.right$.   Let $left = L_{i-1}.right$.
\end{compactenum}
\item The rightmost point of $H_{L_{i-2}.h}$ is beyond $L_{i-2}.right$, i.e., $H_{L_{i-2}.h}.right > L_{i-2}.right$.  Let $left = L_{i-2}.right$.
\end{compactenum}

Let $L_i = \{left, left + l, h\}$ where $H_h.left  \leq L_i.right$ and $H_{h+1}.left > L_i.right$. The inductive hypothesis holds since the segment $[0, L_{i-1}.right]$ is double cover and the segment $[0, L_i.right]$ is single covered.

$C$ accepts a strong double $l$-lid cover if $C$ is fully covered, i.e.,  $L_k.right \geq 1$ and all the high priority segments are in the segment $[0,  L_{k-1}.right]$, i.e., $H_n.right \leq L_{k-1}.right$.

Regarding the time complexity, computing the index $h$ for lid $i$ takes $O(L_i.h - L_{i-1}.h)$.  Hence, computing $h$ for all  lids takes  $O(n)$. Further, computing all lids takes time $O(k)$. Therefore, the time complexity is  $O(\max(n, k))$. 
\end{proof}

The pseudo-code of Lemma~\ref{lem:alg1} is presented in Algorithm~\ref{alg:acceptdouble} .
Similarly, given  $k$ and the lid length we can decide whether $C$ accepts a single lid cover with $k$ lids in time $O(\max(k, n))$. 

\begin{lemma} \label{lem:acceptsdouble}
There exists an algorithm that decides if $C$ accepts a right shifted single $l$-lid cover with $k$ lids in time $O(\max(k, n))$ where $n$ is the number of high priority sections.
\end{lemma}

\begin{proof}
For the proof, we consider  one sentinel high priority sections $H_0$ where $H_0.left = H_0.right = 0$. We also consider one  sentinel lid $L_0$ where $L_0.h = 0$. The proof is constructive. Suppose that at step $i$, the segment $[0, L_{i-1}.right]$ is being covered with $i-1$ lids of length $l$. Therefore, if 
\begin{compactenum}
\item $H_{L_{i-1}.h}.right > L_{i-1}.right$. Let $L_i.left = L_{i-1}.right$. 
\item $H_{L_{i-1}.h}.right \leq L_{i-1}.right$.  Let $L_i.left = H[H_{L_{i-1}.h + 1}.left$. 
\end{compactenum}
Clearly the segment $[0, L_i.right]$ is being covered with $i$ lids. Inductively, the lemma follows.

$C$ accepts a  single $l$-lid cover  if $H_n.right \leq L_k.right$. Regarding the time complexity, observe that computing $L_i.h$ takes at most $O(L_i.h -  L_{i-1}.h)$ time. Therefore, the time complexity is  $O(\max(n, k))$. 
\end{proof}

\begin{algorithm}[ht!]
  \caption{Computes  right shifted  single $l$-lid cover with $k\geq 1$ lids if $C$ 
  accepts it}\
    \label{alg:acceptcover1}
    \TitleOfAlgo{LidCover}

  \SetKwInput{InOut}{Input/Output}        
        \KwIn{$H$: The set of disjoint high priority segments}
        \KwIn{$k$: The number of lids}    
    \KwIn{$l$: Lid length}            
        \KwOut{$True$:  If $C$ accepts a single $l$-lid cover}

       \For{$i=1 \to k$}
       {
           \If{$H_{L_{i-1}.h}.right > L_{i-1}.right$}
           {  
            $left \gets L_{i-1}.right$ \;
        }       
           \Else{
               $left \gets H_{L_{i-1}.h+1}.left$\;
        }
         $L_1 = \{left, left + l, h : H_h.left  \leq left + l$ { \bf and } $H_{h+1}.left > left + l\}$\;
       }
       \Return  $H_n.right \leq L_k.right$\;
\end{algorithm}

The proof and the pseudo-code (Algorithm~\ref{alg:acceptcover1}).
Recall that a block $B$ is the set of lids $\{\ell_{a(1)},\ell_{a(2)}, ..., \ell_{a(b)}\}$ such that $Right(\ell_{a(i)}) = Left(\ell_{a(i+1)})$. We say that a (strong double) $l$-lid cover is feasible if there exists a block  $B$ such that $Left(B)$ and $Right(B)$ are the endpoints of high priority sections. Next we show that if $C$ admits a (strong double) $l$-lid cover, then we can obtain in linear time a new lid length $l' \leq l$ such that $C$ admits a feasible (strong double) $l'$-lid cover.

\begin{lemma} \label{lem:feasible}
Suppose $C$ admits a strong double $l$-lid cover. There exists an algorithm that computes the lid length $l'$ such that $C$ admits a feasible strong double $l'$-lid cover in $O(n)$ time such that $l' \leq l$.
\end{lemma}

\begin{proof} 
Let  $\mathcal{B}$ be  the set of the maximal blocks in a right shifted (strong double) $l$-lids cover. Let $\ell_i$ be the $i$-lid in the block $B$ and let 
$$
m^B(\ell_i)  = 
\begin{cases}
\frac{\ell_i.right - H_{\ell_i.h}.right}{i} &\mbox{ if } \ell_i.right \geq H[\ell_i.h].right \\
\frac{\ell_i.right - H_{\ell_i.h}.left}{i} &\mbox{ if } \ell_i.right < H[\ell_i.h].right\\
\end{cases}
$$
and let $$m = \min_{(\forall B \in \mathcal{B})}\min_{(\forall \ell_i \in B)}(m^B(\ell_i)).$$

Let $l' = l - m$. Now we show that  $C$ accepts a strong double $l'$-lid cover with at least one feasible block. Let $\mathcal{B}'$ denote the  blocks with lid length $l'$. We assume that $\mathcal{B}$ and $\mathcal{B}'$ are sorted.

We show that for all  $B_j \in \mathcal{B}$ and $B'_j \in \mathcal{B}'$ $Left(B_j) = Left(B'_j)$ and  for  each lid $\ell_i \in B_j$ such that 
\begin{compactitem}
\item $\ell_i.right \geq H_{\ell_{i}.h}.right$, we show that 
$$H_{\ell_{i}.h}.right \leq \ell'_i.right  = \ell_i.right - im.$$ 
\item  $\ell_i.right < H_{\ell_{i}.h}.right $, we show  that
$$H_{\ell_{i}.h}.left \leq \ell'_i.right =  \ell_i.right - im.$$
\end{compactitem}

We show for a block $B_j$ by induction on $j$. For the base case assume $j=1$. Trivially,  $Left(B_1) = Left(B'_1) = 0$.  Let $\ell_i \in B_1$ and $\ell'_i \in B'_1$.

If $\ell_i.right \geq H_{\ell_{i}.h}.right$, then $\ell_i.right - im \geq H_{\ell_i.h}.right$ since $$m \leq \frac{\ell_i.right - H_{\ell_i.h}.right}{i}.$$ 

If $\ell_i.right < H_{\ell_{i}.h}.right$ then $H_{\ell_{i}.h}.left \leq \ell_i.right - im$ since  $$m \leq \frac{\ell_i.right - H_{\ell_i.h}.left}{i}.$$

Assume that $Left(B_{j-1}) = Left(B'_{j-1})$ for some $j > 1$. Therefore, for  each lid $\ell_i \in B_{j-1}$,  
\begin{compactitem}
\item if $\ell_i.right \geq H_{\ell_{i}.h}.right$, $$H_{\ell_{i}.h}.right \leq \ell'_i.right  = \ell_i.right - im.$$ Therefore,  $Left(B_j) = Left(B'_j) = H_{\ell_{i}.h+1}.left$ and  each $H_{l}$ where $l \leq  \ell_{i}.h$ is double cover.
\item if  $\ell_i.right < H_{\ell_{i}.h}.right $, $$H_{\ell_{i}.h}.left \leq \ell'_i.right =  \ell_i.right - im.$$ Observe that $R(B'_j) = R(B_j)$ since $m \leq \frac{\ell_l.right - H_{\ell_{l}.h}.right}{i}$ where $\ell_l \in B_j$ is the rightmost lid. 
\end{compactitem}
 
Observe that at least one block is feasible.  It is not difficult to see that the running time is linear on the number of lids since the blocks are disjoint. 
\end{proof}

\begin{algorithm}[ht!]
  \caption{Given a set of maximal blocks of a strong double $l$-lid cover, it returns $m$ where 
  $C$ accepts a feasible strong double $(l-m)$-lid cover  }
  \label{alg:feasabledoublesolution}
   \SetKwInput{InOut}{Input/Output}
       \TitleOfAlgo{FeasibleSolution}
        
        \KwIn{$\mathcal{B}$: The set of maximal blocks of a strong double $l$-lid cover }
        \KwOut{$m$: such that $C$ accepts a feasible strong double $(l-m)$-lid cover}
  
   $m = \infty$\;
        \For{$B \in \mathcal{B}$}
        {
            \For{$\ell_i \in |B|$}
            {
                \If{$\ell_i.right > H_{\ell_i.h}.right$}
               {
                   $m = \min(m, (\ell_i.right - H_{\ell_i.h}.right)/i)$ \;
                    }
            \Else
            {
                  $m = \min(m, (\ell_i.right - H_{\ell_i.h}.left)/i)$ \;
            }
        }
       }
       \Return $m$\;
\end{algorithm}

The pseudo-code is presented in Algorithm~\ref{alg:feasabledoublesolution}). 
The case of the single cover is analogous and is left to the reader as an exercise. Given a strong double $l$-lid cover, we can obtain a feasible solution in linear time and use it in a binary search
to obtain an exact lid length in time $O(n \log n)$.

\begin{theorem}\label{thm:mainalg}
There exists an algorithm that finds the minimum length $\Lambda_{2k}$ such that $C$ accepts a strong double $\Lambda_{2k}$-lid cover in $O(\max(n, k)\log n)$ time.
\end{theorem}

 \begin{proof} 
Observe that $1/(2k)  \leq \Lambda_{2k} \leq 1/k$ since $C$. We use these values to perform a binary search. In other words,  Let $upper =  1/k$,  $lower = 1/(2k)$ and $l = (upper + lower) / 2$. If $C$ does not admit a strong double $l$-lid cover (using Lemma~\ref{lem:alg1}) then the solution is in the interval $upper, mid$. Otherwise,  the solution is in  $mid, lower$. Observe that after enough number of rounds, the result is an approximation. To guarantee an exact solution, when $C$ accepts a strong double $l$-lid cover, we obtain $l'$ such that $C$ accepts a feasible strong double $l'$-lid cover using Algorithm~\ref{alg:feasabledoublesolution}. To check whether $l'$ is the optimal we check if $C$ accepts a strong double $(l' - \epsilon)$-lid cover and return $l'$ if $C$ does not  accept it.

For the time complexity, from Lemma~\ref{lem:acceptsdouble}, checking and computing the strong double lid cover takes linear time. We can compute the maximal blocks greedily in linear time since they are mutually exclusive. Obtaining a feasible solution also takes linear time from Lemma~\ref{lem:feasible}. The theorem follows since there are at most $n^2$ critical blocks and each step it reduces the critical blocks by half.  
\end{proof}

\begin{algorithm}[ht!]
  \caption{Computes the strong double $\Lambda_2k$-lid cover}
  \label{alg:mindoublelid}
  
         \KwIn{$H$: The set of disjoint high priority segments}
        \KwIn{$k$: The number of lids}    
  $upper = 1/k$\;
    $lower = 1/(2k)$\;

   \While{$true$}
  {
       $l \gets (upper+ lower)/2$\;    
     Let $L$ be the set of $l$ lids in $C$\;
     \If{$L$ strong double cover $C$ }{
        Let $\mathcal{B}$ be the set of maximal blocks\;
        Compute $l'$ such that $C$ accepts a feasible strong double $l'$-lid cover \;
        $upper \gets l'$\;
    
        \If{$C$ does not accept a feasible strong double $(l' - \epsilon)$-lid cover }
        {
            \Return $upper$\;
        }
     }
     \Else{
         $lower \gets l$\;
    }          
   }
   \Return $l$\;
 \end{algorithm}

The pseudo-code is presented in Algorithm~\ref{alg:mindoublelid}.

\section{Conclusion}~\label{sec:conclusion}
We study the problem of patrolling a unit line segment that consists of high priority segments that require the maximum protection and low priority segments that require to be visited infinitely often. 
We provide lower and tight upper bounds when $k$ robots are available using the concept of strong double lid cover and single lid cover. We also provide a $O(\max(n, k)\log n)$ running time algorithm that finds the strong double lid cover and single lid cover with optimal lid length. Future work includes different topologies such as trees and graphs as well as distributed online strategies \footnote{Many thanks to Joshua Sack for the valuable discussion.}.


\begin{thebibliography}{1}

\bibitem{chuangpishit2018patrolling}
Huda Chuangpishit, Jurek Czyzowicz, Leszek Gasieniec, Konstantinos Georgiou,
  Tomasz Jurdzinski, and Evangelos Kranakis.
\newblock Patrolling a path connecting a set of points with unbalanced
  frequencies of visits.
\newblock In {\em International Conference on Current Trends in Theory and
  Practice of Informatics}, pages 367--380. Springer, 2018.

\bibitem{collins2013optimal}
Andrew Collins, Jurek Czyzowicz, Leszek Gasieniec, Adrian Kosowski, Evangelos
  Kranakis, Danny Krizanc, Russell Martin, and Oscar Morales~Ponce.
\newblock Optimal patrolling of fragmented boundaries.
\newblock In {\em Proceedings of the twenty-fifth annual ACM symposium on
  Parallelism in algorithms and architectures}, pages 241--250. ACM, 2013.

\bibitem{czyzowicz2017patrolmen}
Jurek Czyzowicz, Leszek Gasieniec, Adrian Kosowski, Evangelos Kranakis, Danny
  Krizanc, and Najmeh Taleb.
\newblock When patrolmen become corrupted: Monitoring a graph using faulty
  mobile robots.
\newblock {\em Algorithmica}, 79(3):925--940, 2017.

\bibitem{czyzowicz2011boundary}
Jurek Czyzowicz, Leszek Gasieniec, Adrian Kosowski, and Evangelos Kranakis.
\newblock Boundary patrolling by mobile agents with distinct maximal speeds.
\newblock In {\em ESA}, pages 701--712. Springer, 2011.

\bibitem{czyzowicz2016patrolling}
Jurek Czyzowicz, Adrian Kosowski, Evangelos Kranakis, and Najmeh Taleb.
\newblock Patrolling trees with mobile robots.
\newblock In {\em International Symposium on Foundations and Practice of
  Security}, pages 331--344. Springer, 2016.

\bibitem{gkasieniec2017bamboo}
Leszek G{\k{a}}sieniec, Ralf Klasing, Christos Levcopoulos, Andrzej Lingas, Jie
  Min, and Tomasz Radzik.
\newblock Bamboo garden trimming problem (perpetual maintenance of machines
  with different attendance urgency factors).
\newblock In {\em International Conference on Current Trends in Theory and
  Practice of Informatics}, pages 229--240. Springer, 2017.

\bibitem{kawamura2012fence}
Akitoshi Kawamura and Yusuke Kobayashi.
\newblock Fence patrolling by mobile agents with distinct speeds.
\newblock In {\em ISAAC}, volume~12, pages 598--608. Springer, 2012.

\bibitem{czyzowicz2014patrolling}
Jurek Czyzowicz, Evangelos Kranakis, Dominik Pajak, and Najmeh Taleb.
\newblock Patrolling by robots equipped with visibility.
\newblock In {\em International Colloquium on Structural Information and
  Communication Complexity}, pages 224--234. Springer, 2014.

\bibitem{czyzowicz2019patrolling}
Jurek Czyzowicz, Kostantinos Georgiou, and Evangelos Kranakis.
\newblock Patrolling.
\newblock In {\em Distributed Computing by Mobile Entities}, pages 371--400.
  Springer, 2019.

\end{thebibliography}

\end{document}